\documentclass[pra,superscriptaddress,nofootinbib,a4paper,twocolumn]{revtex4}
\usepackage{amsmath}
\usepackage{amsfonts}
\usepackage{graphicx}
\usepackage{color}

\newcommand{\mean}{d}

\def\<{\langle}
\def\>{\rangle}
\def\be{\begin{equation}}
\def\ee{\end{equation}}
\newtheorem{theorem}{Theorem}

\newtheorem{lemma}[theorem]{Lemma}

\newtheorem{proposition}[theorem]{Proposition}

\newenvironment{proof}[1][Proof]{\noindent\textbf{#1.} }{\ \rule{0.5em}{0.5em}}

\begin{document}

\title{Reconstruction of Gaussian quantum mechanics from Liouville mechanics with an epistemic restriction}
\author{Stephen D. Bartlett}
\affiliation{School of Physics, The University of Sydney, Sydney, New South Wales 2006,
Australia}
\author{Terry Rudolph}
\affiliation{Controlled Quantum Dynamics Theory, Imperial College London, London SW7
2BW, United Kingdom}
\author{Robert W. Spekkens}
\affiliation{Perimeter Institute for Theoretical Physics, 31 Caroline St.~N, Waterloo, Ontario, Canada, N2L 2Y5}
\date{10 July 2012}

\begin{abstract}
How would the world appear to us if its ontology was that of classical mechanics but every agent faced a restriction on how much they could come to know about the classical state?  We show that in most respects, it would appear to us as quantum. The statistical theory of classical mechanics, which specifies how probability distributions over phase space evolve under Hamiltonian evolution and under measurements, is typically called \emph{Liouville mechanics}, so the theory we explore here is Liouville mechanics with an epistemic restriction.
The particular epistemic restriction we posit as our foundational postulate specifies two constraints.   The first constraint is a classical analogue of Heisenberg's uncertainty principle -- the second-order moments of position and momentum defined by the phase-space distribution that characterizes an agent's knowledge are required to satisfy the same constraints as are satisfied by the moments of position and momentum observables for a quantum state.  The second constraint is that the distribution should have maximal entropy for the given moments.
Starting from this postulate, we derive the allowed preparations, measurements and transformations
%that arise in this theory,
and demonstrate that they are isomorphic to those allowed in Gaussian quantum mechanics and generate the same experimental statistics.
We argue that this reconstruction of Gaussian quantum mechanics constitutes additional evidence in favour of a research program wherein quantum states are interpreted as states of incomplete knowledge,
and that the phenomena that do \emph{not} arise in Gaussian quantum mechanics provide the best clues for how one might reconstruct the full quantum theory.
%The epistemic restriction constrains not only the states; it also constrains the measurements and transformations in order to guarantee that when these are applied to valid states, an agent still does not acquire too much knowledge.
% that this theory is operationally equivalent to Gaussian quantum mechanics insofar as the preparations, measurements and transformations it allows
%For instance, just as one cannot simultaneously predict the outcome of measurements of position and momentum observables in quantum mechanics, one cannot simultaneously know the values of position and momentum.

%We introduce a toy theory of classical particle mechanics supplemented by a constraint on the amount of information an observer may have about the motional state (i.e., the point in phase space).  The resulting theory -- one of Liouville mechanics with an epistemic restriction -- reproduces a wide variety of qualitative features of quantum theory for degrees of freedom that are continuous and dynamical. We develop the operational formulation of this theory by deriving the consequences of this ``classical uncertainty principle'' on state preparations, measurements, and dynamics, and prove that the resulting theory is equivalent to a \emph{subset} of quantum theory, namely, the quantum theory of so-called \emph{Gaussian states}.  Thus, we have shown that a subset of quantum theory (Gaussian quantum mechanics) can be derived from a classical theory (Liouville mechanics) with sole additional axiom being a constraint on what can be known in the theory.
\end{abstract}

\maketitle
%\tableofcontents

%\newpage

\section{Introduction}

What is innovative about quantum mechanics from the perspective of classical
physics?  The thesis we shall defend in this article is that a large part
of quantum mechanics can be understood as arising from a single
innovation relative to classical theories: there is a restriction on how
much any agent can know about the physical state of a classical system.  To be a bit
more precise, the claim is that if one begins with the \emph{statistical}
classical theory, which is to say the one that quantitatively describes an
agent's knowledge of a classical system and therefore specifies how
probability distributions over the classical state space evolve over time
and how they are updated in the course of measurements, and if one then
assumes as a new fundamental postulate that agents are restricted in the
sorts of knowledge they can have about the classical state (or equivalently,
the form of the probability distributions they can prepare), then one can
derive a large part of quantum mechanics in the sense of reproducing its
operational predictions.

We shall consider classical particle mechanics here.  The statistical
theory in this case is known as \emph{Liouville mechanics}. The restriction
on knowledge that we adopt, and which we shall refer to as the \emph{epistemic restriction}, is inspired by Heisenberg's uncertainty principle
together with a principle of entropy maximization.  We derive which
preparations, measurements and transformations are consistent with the
epistemic restriction.  The result is a theory that we refer to as \emph{epistemically-restricted Liouville mechanics}, or ERL mechanics.  We then demonstrate its
equivalence to a subtheory of quantum mechanics which we call \emph{Gaussian
quantum mechanics} (about which we shall say more in a moment).  Significantly, this
implies that all phenomena arising in Gaussian quantum mechanics can be
interpreted in terms of probability distributions over a classical phase
space.  ERL mechanics provides a noncontextual hidden variable model for Gaussian quantum mechanics\footnote{The model is noncontextual in the generalized sense defined in \cite{Spe05}.}.

Within this model, all quantum states are represented by probability distributions that cover a non-vanishing volume of the phase space.  Consequently, they correspond to states of incomplete knowledge about the physical state of the system.  Furthermore, non-orthogonal quantum states correspond to overlapping probability distributions. Consequently, many distinct quantum states are consistent with the system being at a particular point in phase space; a change in the quantum state need not imply a change in reality.  Theories of this sort have been described as $\psi$\emph{-epistemic}~\cite{Spe07,Har10}.  ERL mechanics therefore provides a $\psi$-epistemic hidden variable
model for Gaussian quantum mechanics. The success of this model in reproducing aspects of quantum theory provides additional evidence in favour of interpretations of quantum theory where quantum states describe states of incomplete knowledge rather than states of reality.

We define Gaussian quantum mechanics in terms of the Wigner representation.  Among pure states, it is well known that a wavefunction has a magnitude with Gaussian profile over configuration space if and only if the state admits of a Gaussian (hence non-negative) Wigner representation, and that these are the \emph{only} pure states that have a non-negative Wigner representation~\cite{Hud74}. We here consider a mixed state to be Gaussian only if it has a Gaussian Wigner representation\footnote{Note that these are not the only mixed states that have non-negative Wigner representations -- mixing Gaussian pure states with a non-Gaussian measure, for instance, can yield such a state \cite{BW95}.}.  For all Gaussian states, the Wigner representation can be interpreted as a probability distribution over phase space.  The ability to interpret the Wigner representation of a state as a probability distribution over phase space is sometimes taken as a condition for classicality.  However, as emphasized in Ref.~\cite{Spe08}, this is not sufficient, because one needs to verify that \emph{the entire experiment}, including the measurements and transformations, admits of a classical explanation.  We take the Gaussian measurements to be those associated with positive operator valued measures (POVMs) all the elements of which have Gaussian (hence non-negative) Wigner representations.  This ensures that they can be interpreted as indicator functions (sometimes called ``response'' functions), which specify the conditional probability of the associated outcome for every classical phase space point. Similarly, we take the Gaussian transformations to be those associated with completely positive maps (CP maps) that also have Gaussian (hence non-negative) Wigner representations, which ensures that they can be interpreted as transition probabilities over the phase space.  To summarize, Gaussian quantum mechanics is defined as the subtheory of quantum mechanics\footnote{Note that, following Ref.~\cite{BartlettRowe}, the  Gaussian subtheory of quantum mechanics can be obtained from full quantum mechanics by applying a constraint to motion (in the sense of Dirac~\cite{Dirac}).} including only those preparations, measurements and transformations that have Gaussian Wigner representations, and as we have just noted (and will explain more carefully in Sec.~\ref{sec:GaussianQM}), all such procedures can be given a classical statistical interpretation.  We prove that ERL mechanics is operationally equivalent to Gaussian quantum mechanics by demonstrating that it reproduces the Wigner representation of the latter.

For those familiar with the Wigner representation, the definition of Gaussian quantum mechanics we are adopting here will seem natural, and the possibility of a classical statistical interpretation of this subtheory of quantum mechanics will come as no surprise.
What is not so obvious, and what it is the purpose of this article to demonstrate, is
that it is possible to \emph{derive} Gaussian quantum mechanics starting
from Liouville mechanics and imposing a restriction on knowledge.  This is
a distinction worth emphasizing. Finding a subtheory of quantum mechanics
admitting a non-negative Wigner representation is primarily an exercise in \emph{interpretation} -- one starts from the quantum formalism and proceeds to find a representation that admits of an interpretation in terms of noncontextual hidden variables. By contrast, showing that one can derive
this subtheory of quantum mechanics starting from classical mechanics and
imposing a restriction on knowledge is primarily an exercise in \emph{axiomatization}\footnote{Indeed, imposing an epistemic restriction on a statistical classical theory can be understood as a novel kind of quantization scheme, although, strictly speaking, it may not deserve the title given that it generally returns only a subtheory of quantum theory or an analogue thereof.}.

Of course, only \emph{part} of quantum mechanics, the Gaussian part, has been derived.   We have claimed above that this constitutes a ``large part'' of quantum mechanics, but a
skeptic may rightfully ask what is meant by this.  Insofar as Gaussian
quantum mechanics admits only quadratic Hamiltonians, it might seem to be a
very \emph{small} (and some might say uninteresting) part of the theory.
We argue that it \emph{does} capture a large part of quantum mechanics in
the sense that it captures a large number of the qualitative phenomena that
are usually highlighted as nonclassical, i.e., those that are usually deemed
to rule out a classical worldview.  It is this sort of counting that we
feel to be significant for the project of understanding what is innovative
about quantum theory from a classical perspective\footnote{And likewise for the project of finding physical principles that have some hope of implying quantum theory; see below.}. And by these lights, the counting is very favorable. ERL mechanics
succeeds at reproducing: (i) most basic quantum phenomena (including the
``usual suspects'' on the list of phenomena
that seem to defy classical explanation), for instance, the existence of
complementary measurements (i.e., that cannot be implemented jointly), the
existence of noncommuting measurements (i.e., where the statistics depend on
the order in which they are implemented), the collapse of the wavefunction,
and the no-cloning theorem; (ii) most of the information-processing tasks
that distinguish quantum theory from classical theories, such as
teleportation, key distribution, quantum error correction, and improvements
in metrology; (iii) a large part of entanglement theory, for instance, the
monogamy of pure entanglement, distillation, deterministic and probabilistic
single copy entanglement transformation, catalysis, etcetera; (iv) a large
part of what might be termed the ``statistical
structure'' of quantum theory, such as the isomorphism
between operations on a system and states on a pair of systems (the
Choi-Jamiolkowski isomorphism \cite{BZbook}), the fact that every mixed state has multiple
convex decompositions into pure states and multiple extensions to a pure
state on a larger system (purifications), the fact that every unsharp
measurement can be considered to be a sharp measurement on a larger system
(the Naimark extension \cite{BZbook}), the fact that every irreversible transformation can
be obtained by a reversible transformation on a larger system (the
Stinespring dilation \cite{BZbook}).

But what is the point of deriving only \emph{part} of quantum theory?  Why is a subtheory of quantum mechanics such as Gaussian quantum mechanics interesting?  We are certainly \emph{not} proposing this theory as an empirical competitor to quantum theory.  It is straightforward to prepare Hamiltonians that are not quadratic in position and momentum and hence to demonstrate the existence of deterministic dynamics that is not part of Gaussian quantum mechanics.
Indeed, there is a large range of quantum phenomena that have been predicted and observed in continuous-variable systems for which a description requires non-Gaussian operations, such as states with negative Wigner functions~\cite{GrangierNWF1,GrangierNWF2,GrangierNWF3,GrangierNWF4,Laiho}.  Rather, such theories are of interest as \emph{foils} to quantum theory.  They depict ways in which the world \emph{might have been}.  This is useful for identifying principles from which one can derive quantum theory because it is only by describing a broad landscape of possible theories that we can specify the sense in which quantum theory is special.
Foil theories that reproduce many quantum phenomena (such as the one considered here) are particularly useful for ruling out possible axiom schemes.  For instance, if one is contemplating a possible axiom scheme and all of the axioms hold true for Gaussian quantum mechanics, then one recognizes immediately that they are not sufficient for deriving the whole of quantum mechanics.  One needs to look at the phenomena that Gaussian quantum mechanics does \emph{not} reproduce in order to find an adequate set of axioms.

The most significant phenomena that are not included in Gaussian quantum mechanics are as follows.  Bell inequality violations~\cite{Bell64} are not included because epistemically-restricted Liouville mechanics provides a local hidden variable model for Gaussian quantum mechanics.  Indeed, locality for bipartite continuous-variable EPR experiments was already established in Refs.~\cite{BellEPR,OuCVBell}. The Kochen-Specker theorem~\cite{KS67} is not included, nor are violations of operational noncontextuality inequalities, as defined in Ref.~\cite{Spe09}.  This follows from the fact that epistemically-restricted Liouville mechanics is a noncontextual hidden variable model, or equivalently because the Wigner representation is a nonnegative quasi-probability representation and any such representation is a noncontextual hidden variable model~\cite{Spe08}.  Exponential speed-up for computation is not included (assuming it exists)~\cite{Nie00}.  This follows from the existence of an efficient classical simulation of Gaussian quantum mechanics~\cite{Bar02}. Quantum interference phenomena and quantization of quantities such as angular momentum and energy are also absent from the theory.   Another example is the phenomenon described in the recent article by Pusey, Barrett and Rudolph~\cite{PBR11}, which can be understood as a violation of a notion of noncontextuality for preparations~\cite{Spe05}.   There is little doubt that a direct test of the predictions of Gaussian quantum mechanics versus the predictions of the full quantum theory on any of these fronts would rule in favour of the full quantum theory.  Proposals for experimental tests of the Bell inequalities for continuous variable systems have been made in Refs.~\cite{Cavalcanti07,He10,Leonhardt95,Gilchrist98,Yurke99,Yurke01}.

An even more informative distinction is between phenomenon that can occur in \emph{some} classical statistical theory with an epistemic restriction (not necessarily classical particle mechanics), and those that cannot. For instance, for phenomena that are characteristic of finite-dimensional quantum systems, the relevant question is whether they arise in statistical theories of \emph{discrete} classical systems, such as the theory considered in Ref.~\cite{Spe07}.  As another example, interference and quantization phenomena may yet be incorporated under the umbrella of epistemically-restricted statistical theories wherein the ontology is \emph{fields} rather than particles.  The phenomena of nonlocality, contextuality and quantum exponential speed-up are distinguished in the list insofar as they are clearly insensitive to the degree of freedom one is considering.

This categorization of phenomena -- into those which arise in classical statistical theories with an epistemic restriction and those which do not -- is a useful application of our results.  From the perspective of the $\psi$-epistemic research program, there are two tiers of nonclassicality: the
first tier contains the phenomena that can be explained merely by
postulating an epistemic restriction but maintaining the notion of an
underlying classical ontology, while the second tier contains the rest.
For the purposes of moving the research program forward, it is the second
tier that is the most interesting, for it is by studying these phenomena
that one can hope to deduce additional principles that might supplement the
epistemic restriction and allow a derivation of the full quantum theory.
Consequently, it is useful to categorize as many phenomena as possible in
order to extend the list of second tier phenomena.

Nonetheless, we feel that the diversity and foundational importance of the
quantum phenomena that can be reproduced in classical
epistemically-restricted theories suggests that there is something right
about this research program.  In particular, its success suggests to us
that there may be an axiomatization of quantum theory of the following sort.
The first axiom states that there is a fundamental restriction on how much
observers can know about systems.  The second embodies some novel principle
about reality (rather than our knowledge thereof).  Ultimately, the first axiom ought to be derivable from the second because what one physical system can know about another ought to be a consequence of the nature of the
dynamical laws.

\subsection{Previous work that is relevant to this article}

The idea that quantum states are states of incomplete knowledge (i.e., \emph{epistemic} states) rather than states of reality (i.e., \emph{ontic} states) is an old one.  In Ref.~\cite{Har10}, it is argued that Einstein was an early advocate of $\psi$-epistemic hidden variable models.
\begin{quote}
\emph{... \strut I incline to the opinion that the wave function does not
(completely) describe what is real, but only a (to us) empirically
accessible maximal knowledge regarding that which really exists [...] This
is what I mean when I advance the view that quantum mechanics gives an
incomplete description of the real state of affairs.}
%\hfill{ --A. Einstein~\cite{einsteinquote}}
\newline --A. Einstein~\cite{einsteinquote}
\end{quote}
E.~T.~Jaynes, famous for his information-theoretic derivation of many results of
classical thermodynamics~\cite{Jay57}, also argued that many results in
quantum theory could be understood in this manner, but that a prerequisite
for doing so is to properly distinguish between ontic and epistemic concepts
in quantum theory.  
\begin{quote}
\emph{...present quantum theory not only does not use -- it does not even
dare to mention -- the notion of a \textquotedblleft real physical
situation.\textquotedblright\ Defenders of the theory say that this notion
is philosophically naive, a throwback to outmoded ways of thinking, and that
recognition of this constitutes deep new wisdom about the nature of human
knowledge. I say that it constitutes a violent irrationality, that somewhere
in this theory the distinction between reality and our knowledge of reality
has become lost, and the result has more the character of medieval
necromancy than of science.}
%\hfill{ --E.T. Jaynes~\cite{jaynesquote}}
\newline --E.T. Jaynes~\cite{jaynesquote}
\end{quote}
\begin{quote}
\emph{But our present QM formalism is not purely epistemological; it is a peculiar mixture describing in part realities of Nature, in part incomplete human information about Nature --- all scrambled up by Heisenberg and Bohr into an omelette that nobody has seen how to unscramble. Yet we think that the unscrambling is a prerequisite for any further advance in basic physical theory. For, if we cannot separate the subjective and objective aspects of the formalism, we cannot know what we are talking about; it is just that simple.}
%\hfill{ --E.T. Jaynes~\cite{jaynesquote2}}
\newline --E.T. Jaynes~\cite{jaynesquote2}
\end{quote}
Ballentine has argued in favour of the
thesis that quantum states describe the statistical properties of a virtual
ensemble of systems, which is equivalent to saying that it describes one's
limited information about a single system drawn from the ensemble~\cite{Bal70,Bal94}.  More
recent work that we take to be indicative of the explanatory power of $\psi$-epistemic hidden variable models for quantum theory are Refs.~\cite{Eme01,Har99,Kir03,Spe07}.

There is also much interest in the notion that quantum states are states of
knowledge outside the context of hidden variable approaches. Given that
\emph{pure} quantum states are the ones with maximum information content
(i.e., the most predictability), if one accepts that even these are
epistemic, one is accepting that \emph{maximal} information is not \emph{complete} information.  This is a notion that has been popular of late as a principle from which quantum theory might be
derived.  For instance, it is central to the quantum Bayesian or ``Q-bist'' research program of Fuchs and his collaborators \cite{CF96, CFS02b, Fuc02, Fuchssamizdat,FuchsQBism}.
The work of Leifer and developments thereof \cite{Lei06,SpeLei11} are also along this vein.
The idea also appears in the context of operational reconstructions of quantum theory,
for instance, in Refs.~\cite{Zei99,BZ99, PDB08}.

It should be noted that researchers who may agree
that quantum states are states of incomplete knowledge may still not agree
on what this knowledge is knowledge \emph{about}.  For instance, in quantum
Bayesianism, it is about the ``outcomes of future
interventions'' on the system rather than about some
pre-existing reality.  Because of Bell's theorem and the Kochen-Specker
theorem, it is clear that if quantum states are states of incomplete
knowledge, this knowledge \emph{cannot} be about local and noncontextual
hidden variables.  But local and noncontextual hidden variables do not
necessarily exhaust the possibilities for something to meet the description
of a pre-existing reality and consequently there may still be room for an
interpretation along these lines.  Indeed, this is the idea of the speculative
axiomatization described above.

The previous work that is most relevant to this article is Ref.~\cite{Spe07}, where an epistemic restriction is applied in the context of a classical theory of systems with discrete state spaces to obtain a ``toy theory'' that is very close -- but not equivalent -- to a subset of quantum theory, namely, the stabilizer formalism for qubits. The present work can be seen as an application to continuous variable systems of the idea proposed there.

\subsection{Structure of the paper}

We review the key features of quantum mechanics and Liouville mechanics in Sec.~\ref{sec:Prelim}, focusing on the properties that are important for our discussion.  We introduce ERL mechanics in Sec.~\ref{sec:ERLTheory}, first with a formulation of the epistemic constraint (Sec.~\ref{sec:EpistemicRestriction}), as well as a pedagogical discussion of the basic features of this theory (Sec.~\ref{sec:BasicFeatures}) along with an analysis of some of the quantum phenomena that it reproduces (Sec.~\ref{sec:phenomena}).  Our main result is presented in Sec.~\ref{sec:Equivalence}, where we prove the operational equivalence of ERL mechanics and Gaussian quantum mechanics.

\section{Preliminaries}

\label{sec:Prelim}

The restriction we adopt is motivated by Heisenberg's uncertainty principle in quantum theory, so we begin with a review of the latter in Sec.~\ref{sec:QuantumIntro}, focussing on the particular elements that will be important for ERL mechanics.  Next, in Sec.~\ref{sec:Liouville}, we review the formulation of Liouville mechanics.

\subsection{Quantum Mechanics}
\label{sec:QuantumIntro}

For a system with a configuration described by $n$ degrees of freedom (e.g., 1 particle in $n$ dimensions, $n$/3 particles in 3 dimensions, $n$ particles in 1 dimension, etcetera), the $2n$ canonical operators for the positions $\{\hat{q}_{i},i=1,\ldots ,n\}$ and corresponding momenta $\{\hat{p}_{i},i=1,\ldots ,n\}$ satisfy $[\hat{q}_{i},\hat{p}_{j}]=i\hbar \delta _{ij}\hat{I}$, with $\hat{I}$ the identity operator. We express the $2n$ canonical operators in the form of phase space coordinates, defining $\hat{z}_{2i-1}=\hat{q}_{i}$ and $\hat{z}_{2i}=\hat{p}_{i}$ for $i=1,\ldots ,n$.  These operators satisfy $[\hat{z}_{i},\hat{z}_{j}]=i\hbar \Sigma _{ij}$, with $\Sigma$ the skew-symmetric $2n\times 2n$ matrix $\Sigma_{ij}=\delta _{i,j+1}-\delta _{i+1,j}$, that is,
\begin{equation}
\label{eq:SymplecticForm}
\Sigma =
\begin{pmatrix}
0 & -1 &  0 & 0 & \dots \\
1 & 0 &  0 & 0 & \\
0 & 0 & 0  & -1 & \\
0 & 0 & 1 & 0 &  \\
 \vdots & &  & & \ddots
\end{pmatrix}
\end{equation}

%This is known as the \emph{symplectic form}.

The state of a quantum system is described by a density operator $\rho$.  For a state $\rho$, the \emph{means} of the canonical operators are defined
to be
\begin{equation}
  \mean_{i}(\rho)=\mathrm{Tr}(\rho \hat{z}_{i}),  \label{means}
\end{equation}
and the \emph{covariance matrix} is defined as
\begin{align}
  \gamma _{ij}(\rho)& =\mathrm{Tr}\bigl(\rho (\hat{z}_{i}-\mean_{i})(\hat{z}_{j}-\mean_{j})\bigr)-i\hbar \Sigma _{ij}  \notag \\
  & =2\mathrm{Re}\,\mathrm{Tr}\bigl(\rho (\hat{z}_{i}-\mean_{i})(\hat{z}_{j}-\mean_{j})\bigr)\,,
  \label{eq:CovarianceMatrix}
\end{align}
where the operator ordering in this definition is chosen such that $\gamma_{ij}$ is Hermitian.  In terms of the covariance matrix, a general form of the quantum uncertainty principle can be expressed as
\begin{equation}
\gamma (\rho )+i\hbar \Sigma \geq 0\,.
\label{eq:CanonicalUncertaintyRelations}
\end{equation}
This can be \emph{derived} from the canonical commutation relations of the operators \cite{GossonTextbook}.

For illustration, we now show that this inequality reduces to the usual uncertainty relation for a single system in 1 dimension (single $\hat{q}$ and $\hat{p}$). We have
\begin{equation}
\gamma (\rho )=
\begin{pmatrix}
2(\Delta q)^{2} & \langle \hat{q}\hat{p}+\hat{p}\hat{q}\rangle -2\langle
\hat{q}\rangle \langle \hat{p}\rangle \\
\langle \hat{q}\hat{p}+\hat{p}\hat{q}\rangle -2\langle \hat{q}\rangle
\langle \hat{p}\rangle & 2(\Delta p)^{2}%
\end{pmatrix}
\,,  \label{eq:Gamma}
\end{equation}
where $(\Delta q)^{2}=\langle (\hat{q}-\langle q\rangle )^{2}\rangle $ and
similarly for $(\Delta p)^{2}$. The condition $\gamma (\hat{\rho})+i\hbar
\Sigma \geq 0$ for a $2\times 2$ matrix is equivalent to $\mathrm{det}
(\gamma (\hat{\rho})+i\hbar \Sigma )\geq 0$. Thus,
\begin{equation}
4(\Delta q)^{2}(\Delta p)^{2}\geq (\langle \hat{q}\hat{p}+\hat{p}\hat{q}
\rangle -2\langle \hat{q}\rangle \langle \hat{p}\rangle )^{2}+\hbar ^{2}\,,
\label{eq:FullUncertaintyRelation}
\end{equation}
which implies the standard form of the quantum uncertainty principle
\begin{equation}
  \Delta q\Delta p\geq \hbar /2\,.
  \label{eq:QuantumUncertaintyPrinciple}
\end{equation}

Because all unitary transformations preserve the commutation relations, they also preserve the general form of the uncertainty relation, Eq.~\eqref{eq:FullUncertaintyRelation}.  In this article, what will be relevant are those unitary transformations that act \emph{linearly} on the canonical operators.  These are the \emph{linear symplectic transformations}.  Each such transformation can be represented by a $2n \times 2n$ real matrix $A$ satisfying
\begin{equation}
  A^\dag \Sigma A = \Sigma\,,
\end{equation}
that acts on the canonical operators as
\begin{equation}
\hat{\mathbf{z}} \rightarrow A^\dag \hat{\mathbf{z}},
\label{eq:SymplecticTrans}
\end{equation}
where $\hat{\mathbf{z}} = (\hat{z}_i)$ is the vector of canonical operators.
%Specifically, a unitary representation $M$ of a linear symplectic transformation $A$ acts on the canonical operators as
%\begin{equation}
%  M(A):\hat{z}_{i}\rightarrow \hat{z}_{i}^{\prime }=\sum_{k=1}^{2n}\hat{z}_{k}a_{ki}\,,
%  \label{eq:SymplecticTrans}
%\end{equation}
%with $A=(a_{ij})$ a real matrix satisfying $A^{\dag }\Sigma A=\Sigma $.
It follows that the action of such a symplectic transformation on the vector of means can be inferred from
Eq.~(\ref{means}) to be simply
\begin{equation}
\mathbf{\mean} \rightarrow A^\dag\mathbf{\mean} \,,
\end{equation}
and it then follows from Eq.~(\ref{eq:CovarianceMatrix}) that the action on the
covariance matrix is
\begin{equation}
\gamma \rightarrow A^{\dag}\gamma A\,.
\end{equation}
Because $A^{\dag }\Sigma A=\Sigma$, the transformed covariance matrix $A^{\dag}\gamma A$ also satisfies Eq.~(\ref{eq:CanonicalUncertaintyRelations}).

\subsection{Liouville mechanics}
\label{sec:Liouville}

Liouville mechanics is the dynamical theory for states of knowledge about a classical system.  A classical system is described by a \emph{phase space}, and the real state of affairs of a classical system, i.e., its \emph{ontic state}, corresponds to a point in phase space.  Recall that a phase space is an even-dimensional differentiable manifold $\mathcal{M}$ with a symplectic structure, meaning that it locally admits coordinates $\{q_{i},p_{i};i=1,\ldots ,n\}$ with a Poisson bracket $\{q_{i},p_{j}\}=\delta _{ij}$~\cite{Arn97}.  We express these $2n$ canonical coordinates in the form $z_{2i-1} = q_i$ and $z_{i}=p_i$ for $i=1,\ldots,n$, and the Poisson bracket in these coordinates is $\{z_i,z_j\}=\Sigma_{ij}$, with $\Sigma$ defined as above.  We use $\mathbf{z}$ to denote the vector of coordinates, defining a point in the phase space $\mathcal{M}$, i.e., $\mathbf{z} \in \mathcal{M}$.  Systems can be combined into composite systems, with a phase space given by the Cartesian product of the phase spaces of the components, $\mathcal{M}_{AB} = \mathcal{M}_A \times \mathcal{M}_B$.
%For any classical system, elementary or composite, we use $\mathbf{z}$ as a shorthand to denote a point in phase space, i.e., an ontic state, giving a precise value to all $\{q_i,p_i;i=1,\ldots,n\}$.

Let $L(\mathcal{M})$ be the space of real-valued functions on the phase space $\mathcal{M}$, that is, $L(\mathcal{M})=\{f:\mathcal{M}\rightarrow\mathbb{R}\}.$  The space of functions on a composite system's phase space is the \emph{tensor product} of the space of functions over the component phase spaces, $L(\mathcal{M}_{AB})=L(\mathcal{M}_{A})\otimes L(\mathcal{M}_{B})$, that is, the closure of the Cartesian product $L(\mathcal{M}_{A})\times L(\mathcal{M}_{B})$ under linear combinations.
A function is non-negative, $f \geq 0$, if $f(\mathbf{z}) \geq 0$ for all $\mathbf{z} \in \mathcal{M}$.  We can define a norm on the set of functions by
\begin{equation}
  |f|=\int_{\mathcal{M}}d\mathbf{z}\, f(\mathbf{z} )\,, \quad f \in L(\mathcal{M}).
\end{equation}

In a classical theory, any probability distribution on phase space, sometimes called a \emph{Liouville distribution}, represents a possible description of an observer's knowledge of that system.  That is, any Liouville distribution is a valid \emph{epistemic state} for the system.  The probability distributions on $\mathcal{M}$ are the functions $\mu$ that
are non-negative with norm $1$.  We define this set to be $L_+(\mathcal{M})$, i.e.,
\begin{equation}
  L_+(\mathcal{M}) = \{ \mu \in L(\mathcal{M})\ \text{s.t.} \ \mu \geq 0,
  |\mu | =1 \}\,.
\end{equation}

A Liouville distribution $\mu \in L_+(\mathcal{M})$ is a probability distribution (strictly speaking, a probability density), with which one can define expectation values of functions $f$ on $\mathcal{M}$ denoted $\langle
f\rangle_{\mu}=\int_{\mathcal{M}}f(\mathbf{z})\mu(\mathbf{z})d\mathbf{z}$. Thus, to every Liouville distribution $\mu$ we assign a set of means $\mean_{i}(\mu)=\langle z_{i}\rangle _{\mu}$ and a covariance matrix
\begin{align}
  \gamma _{ij}(\mu)& =2\bigl\langle(z_{i}-\mean_{i}(\mu))(z_{j}-\mean_{j}(\mu))\bigr\rangle_{\mu} \notag \\
  & =2\langle z_{i}z_{j}\rangle _{\mu}-2\langle z_{i}\rangle _{\mu}\langle
  z_{j}\rangle _{\mu}\,,
  \label{eq:CovarianceMatrixLiouville}
\end{align}
to the canonical coordinates $\{z_i; i=1,\ldots,2n\}$.  For any Liouville distribution, the covariance matrix is positive semidefinite
\begin{equation}
  \gamma (\mu)\geq 0\,.
  \label{eq:LiouvilleCovarianceCondition}
\end{equation}
In the case of a single system in 1 dimension, we have
\begin{equation}
  \gamma (\mu)=
  \begin{pmatrix}
  2(\Delta q)^{2} & 2\left( \left\langle qp\right\rangle -\langle q\rangle
  \langle p\rangle \right) \\
  2\left( \left\langle qp\right\rangle -\langle q\rangle \langle p\rangle
  \right) & 2(\Delta p)^{2}
  \end{pmatrix}
  \,.
\end{equation}
Eq.~(\ref{eq:LiouvilleCovarianceCondition}) yields no restriction on the product of variances of position and momentum except for the trivial one,
\begin{equation}
  \Delta q\Delta p\geq 0.
\end{equation}
It is the $i\hbar \Sigma $ term that appears in Eq.~(\ref{eq:CanonicalUncertaintyRelations}) and that is absent from Eq. (\ref{eq:LiouvilleCovarianceCondition}) which accounts for the existence of a restriction on the product of the variances in position and momentum in quantum theory and the absence of any such restriction in Liouville mechanics.

%\textbf{Say that symplectic transformations preserve the LHS of uncertainty relation.  Reference back to this later when we say that only the max-ent principle requires us to restrict our attention to linear symplectic transformations.}

\section{Epistemically-restricted Liouville mechanics}
\label{sec:ERLTheory}

\subsection{The epistemic restriction}
\label{sec:EpistemicRestriction}

What we shall consider in this paper is a theory that can be obtained from Liouville mechanics by adding a foundational postulate,  a restriction on the allowed epistemic states (phase-space distributions) within the theory.
\medskip

%\begin{description}
%[Epistemic Restriction.]
\noindent\textbf{Epistemic Restriction.}
A distribution over phase space, $\mu\in L_+(\mathcal{M})$, can describe
an observer's knowledge of the ontic state of a physical system if and only if it satisfies both of the following constraints\footnote{Exceptions to this rule arise if the evidence upon which the observer's knowledge is conditioned is related to the physical system in a non-standard way, for instance, from a pair of measurements on the system, one in the past and the other in the future (pre- and post-selection)~\cite{Aharonov2008}. We discuss this caveat at the end of this section.}:
\begin{description}
\item[(a) The classical uncertainty principle (CUP).] The covariance matrix of the distribution, $\gamma(\mu)$, must satisfy the inequality
\begin{equation}
  \gamma (\mu)+i\lambdabar \Sigma \geq 0\,,  \label{eq:EpistemicRestriction}
\end{equation}
where $\lambdabar > 0$ is a free parameter of the theory (with units of action).
\item[(b) The maximum entropy principle (max-ent).] The distribution $\mu$ must have maximum entropy over the phase space,
\begin{equation}
  S(\mu)=-\int_{\mathcal{M}} \mu(\mathbf{z})\mathrm{\log }\mu(\mathbf{z})\mathrm{d}\mathbf{z}\,,
\end{equation}
among all possible phase-space distributions with the same covariance matrix.
\end{description}

\noindent If a distribution $\mu \in L_+(\mathcal{M})$ satisfies both the CUP and the max-ent condition, we say that it is a \emph{valid} epistemic state, i.e., $\mu \in L_{\rm valid}(\mathcal{M})$ where $L_{\rm valid}(\mathcal{M}) \subset L_+(\mathcal{M})$ is the set of all valid epistemic states on $\mathcal{M}$.

The CUP is obviously chosen to parallel the quantum uncertainty principle, Eq.~(\ref{eq:FullUncertaintyRelation}), with $\lambdabar$ playing the role of Planck's constant.  For a single system in 1 dimension, the CUP implies
\begin{equation}
4(\Delta q)^{2}(\Delta p)^{2}\geq (\langle \hat{q}\hat{p}+\hat{p}\hat{q}
\rangle -2\langle \hat{q}\rangle \langle \hat{p}\rangle )^{2}+\lambdabar ^{2}\,,
\label{eq:ClassicalFullUncertaintyRelation}
\end{equation}
which in turn implies
\begin{equation}
  \Delta q\Delta p\geq \lambdabar /2\,,
  \label{eq:ClassicalUncertaintyPrinciple}
\end{equation}
where variances and expectation values are relative to the classical distribution over phase space.  So we see that the free parameter $\lambdabar$ fixes the minimum product of variances of position and momentum in our theory.  Note that the presence of the imaginary number $i$ in the CUP does not imply that we have made some kind of transition from probability distributions to complex amplitudes; the inequality (\ref{eq:EpistemicRestriction}) represents a set of inequalities on the real eigenvalues of the Hermitian matrix $\gamma (\mu)+i\lambdabar \Sigma$  and is therefore simply a compact way of expressing a set of constraints on the variances and cross-correlations.

%It is also useful to note that if the classical uncertainty principle holds for a Liouville
%distribution $\rho$ on $\mathcal{M}$, then any distribution $\mu'$ obtained from $\rho$ by Hamiltonian evolution will also satisfy the principle. Noting that Hamiltonian evolution results in a symplectic
%transformation $A$ on $\mathcal{M}$, the proof is identical to that provided in the quantum case.  \sdb{This is not quite correct, as the quantum case held only for linear symplectic transformations.  In what follows, we will only want to consider the linear symplectic transformations anyway, because they preserve the max-ent condition.  I suggest we delete this paragraph, and let the point be made by the `linear symplectic transformations' theorem of the next section.}

There are many distributions $\mu$ that have a given set of mean values of
the canonical coordinates, $\mathbf{\mean}$, and covariance matrix $\gamma$. The max-ent part of our
epistemic restriction specifies that among these, the only distribution that an agent can assign to a system is the one that maximizes the entropy of the distribution over the phase space for this set of mean values and covariance matrix. According to Jaynes' max-ent principle~\cite{Jay57}, this assumption ensures that an agent should have the maximum uncertainty about the physical state of the system consistent with knowing the means and the covariance matrix. The max-ent constraint is ultimately justified \emph{a posteriori} -- we assume it because the theory that one derives without it is less analogous to quantum theory.  In particular, while more distributions would be allowed in a theory that did not assume the max-ent condition, they have no counterpart in quantum theory. Moreover, these additional distributions come at a cost, namely, that the set of allowed measurements is highly proscribed relative to the theory that does assume the max-ent condition.  See Appendix~\ref{sec:Whymaxent} for more details.

It can be shown that the set of
distributions that satisfy \emph{both} the CUP and the max-ent condition are multi-variate Gaussians, given by
\begin{equation}
\mu (\mathbf{z})=\frac{1}{(2\pi)^{n}{\rm det}{\gamma}^{1/2}}\mathrm{exp}
\left( -\frac{1}{2}(\mathbf{z}-\mathbf{\mean})^{T}{\gamma}^{-1}(\mathbf{z}-\mathbf{\mean})\right)\,,
\label{eq:GaussianDistFromCovMatrix}
\end{equation}
where $\gamma$ is the covariance matrix and $\mathbf{\mean}$ is the vector of
mean values of the coordinates.  (The analogy with the Wigner
function will be explored in Sec.~\ref{sec:Wigner}.)  Note that if $\gamma(\mu)$ is not strictly positive-definite, one is required to use a pseudoinverse $\gamma^{-1}$ in this expression.

The theory of ERL mechanics describes a world that is classical in its ontology but wherein there is a fundamental restriction on experimental operations, that is, a restriction on what sorts of preparations, measurements and transformations are possible such that an observer's knowledge of a system must always be given by a probability distribution $\mu$ that satisfies the epistemic restriction.  Whereas it is often argued that one cannot interpret the quantum uncertainty principle, Eq.~(\ref{eq:QuantumUncertaintyPrinciple}), as expressing a constraint on what one knows about well-defined and pre-existing values of the position and momentum\footnote{For this reason, the term \emph{indeterminacy relation} is sometimes argued to be preferable to \emph{uncertainty relation} in the quantum context.}, this is precisely the physical content of the \emph{classical} uncertainty principle, Eq.~(\ref{eq:EpistemicRestriction}).

Finally, we emphasize that the epistemic constraint has implications both for predictions as well as retrodictions within ERL mechanics, and it is worth taking note of a subtlety in this regard.  In quantum theory, the Heisenberg uncertainty principle applies for pure predictions and pure retrodictions, and not for inferences based on pre- and post-selection \cite{Aharonov2008}.  To see this, consider a sequence of three von-Neumann measurements on a completely mixed state.  The first is a position measurement with outcome $q$ (the pre-selection), the last is a momentum measurement finding outcome $p$ (the post-selection), and the intermediate measurement is the one whose outcome is to be estimated.  There is no uncertainty relation because there is no trade-off between the certainty one has about the outcome of an intermediate position measurement and the certainty one has about the outcome of an intermediate momentum measurement.  If the intermediate measurement is of position, then one knows that its outcome will be $q$, based on the pre-selection, while if it is of momentum, then one knows that its outcome will be $p$, based on the post-selection.  One is certain of the outcome in both of the counterfactual scenarios.  In ERL mechanics, one can also come to learn both the position and the momentum of a physical system using pre- and post-selection.  In other words, the epistemic restriction, like the uncertainty principle in quantum theory, applies only for pure predictions and pure retrodictions, and not for inferences based on pre and post-selection.

\subsection{Basic features of ERL mechanics}
\label{sec:BasicFeatures}

%With the epistemic restriction as our sole additional axiom to classical Liouville theory, we can now explore some its consequences.
In this section, we will describe some of the basic features of ERL mechanics, with an emphasis on the qualitative rather than formal descriptions.  In Sec.~\ref{sec:phenomena}, we illustrate in detail how several paradigmatic quantum phenomena are reproduced.
%Finally, in Sec.~\ref{sec:Characterisation}, we formally characterise the operational features of ERL mechanics.

%In this section, we characterize ERL mechanics. In Sec.~\ref{sec:phenomena}, we describe some of the basic features of ERL mechanics, with an emphasis on the qualitative rather than formal descriptions.  Finally, in Sec.~\ref{sec:Characterisation}, we formally characterise the operational features of ERL mechanics.

%Specifically, starting with the valid epistemic states in ERL mechanics, we will define and characterise the valid transformations and measurements as those which always preserve the epistemic restriction.  Thus, from only the epistemic restriction, we will show that ERL mechanics forms a consistent theory, which we then show to be equivalent to a subset of quantum mechanics -- Gaussian mechanics -- with all of its associated `quantum' phenomena.

\subsubsection{Reversible transformations}
\label{sec:Reversible}
%\subsubsection{Linear symplectic transformations}

Every transformation between epistemic states must be the result of a transformation of the ontic states. The reason is that we are contemplating a world that obeys classical dynamics, so that \emph{by assumption} dynamics corresponds to a mapping of the ontic state space to itself.  If an agent lacks knowledge of this dynamics, then they might describe what they know by a stochastic map, determining a probability distribution over final ontic states for every initial ontic state.  However, if the agent can reverse the transformation, then it follows that this map must be a bijective function over the ontic state space.  The set of reversible transformations on canonical coordinates that are allowed in classical mechanics are the symplectic transformations~\cite{Arn97}.  These are precisely the transformations generated by time evolution under an arbitrary Hamiltonian (that is, one which is an arbitrary function of the canonical coordinates).  Liouville's Theorem~\cite{Arn97} tells us that phase space volumes are preserved under such symplectic transformations.  Thus, in that any covariance matrix can be viewed as defining a volume of phase space via an ellipsoid with axes given by the eigenvectors of the covariance matrix, satisfaction of the CUP is preserved by all symplectic transformations.  That is, if the CUP is satisfied and a symplectic transformation is applied to the system, then it continues to be satisfied.

%\textbf{Describe Liouville's theorem and point out that the LHS of uncertainty relation is a measure of phase space volume and is consequently preserved by all symplectic transformations.  Therefore, the CUP is preserved by all classical dynamics.}
In contrast to the CUP, the max-ent condition is only preserved by a subset of the symplectic transformations, namely, the \emph{linear} symplectic transformations.
%The max-ent condition will be preserved only by linear transformations on the phase-space variables, as
%only linear transformations map all Gaussian functions to Gaussian functions.
%Thus, it is the linear symplectic transformations which preserve both the CUP and max-ent condition,
%As already noted in Sec.~\ref{sec:QuantumIntro},
These are defined as the symplectic transformations that act linearly on the canonical coordinates.  Each such transformation can be represented by a $2n \times 2n$ real matrix $A$ satisfying $A^\dag \Sigma A = \Sigma$, where $\Sigma$ is defined in Eq.~\eqref{eq:SymplecticForm}, and that acts on the symplectic vector space $\mathcal{M}$ as
\begin{equation}
\mathbf{z} \rightarrow A^\dag \mathbf{z}\,.
\end{equation}
 The linear symplectic transformations are those that can be generated by time evolution under a Hamiltonian at most quadratic in the canonical coordinates.  These are the only symplectic transformations that preserve the max-ent condition because these are the only ones that map all Gaussian functions to Gaussian functions.  Such transformations correspond to phase-space displacements, rotations and squeezing.
 %As their action on the canonical coordinates is linear, such transformations can be described by the equation $\mathbf{z} \rightarrow A^\dag \mathbf{z}$ for a $2n \times 2n$ real matrix $A$.  For the linear transformation $A$ to be symplectic, it must preserve the symplectic form,
%\begin{equation}
%  A^\dag \Sigma A = \Sigma\,.
%\end{equation}
We noted that any symplectic transformation preserves the CUP, so linear symplectic transformations must as well.  Nonetheless, it is illustrative to see a direct proof of this fact. It suffices to note that a linear symplectic transformation $A$ induces a transformation of the covariance matrix of the form $\gamma \to A\gamma A^{\dag}$, so that if $\gamma +i\hbar\Sigma$ is a positive matrix, then so is $A(\gamma +i\hbar\Sigma)A^{\dag}$, and this in turn implies that $A\gamma A^{\dag} +i\hbar\Sigma$ is positive.  Fig.~\ref{fig:NoDecrease} illustrates a transformation on a valid Liouville distribution that is not allowed within ERL mechanics.

\begin{figure}
\begin{center}
\includegraphics[width=1\hsize]{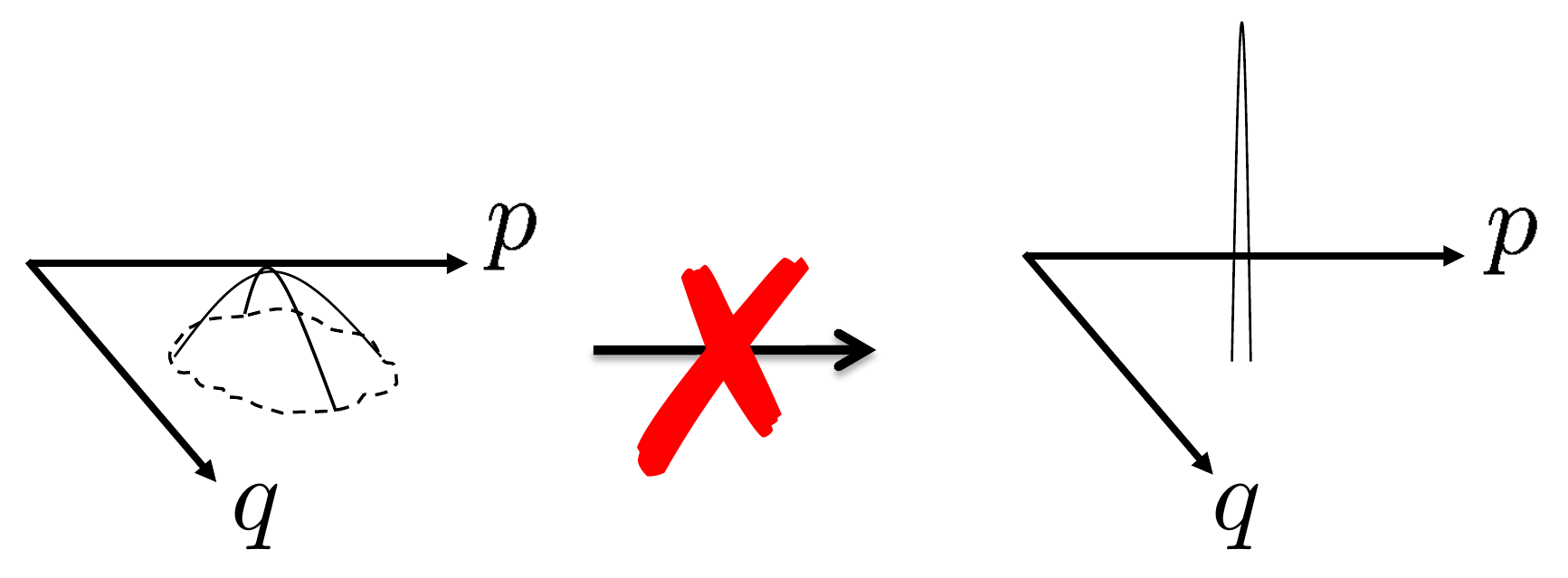}
\caption{The reversible transformations within ERL theory preserve the epistemic constraint.  As a result, a transformation that reduces the uncertainty in both position and momentum of a Gaussian distribution, as illustrated, is not allowed within ERL mechanics.}\label{fig:NoDecrease}
\end{center}
\end{figure}

%\begin{lemma}[Linear symplectic transformations]\label{linsymplectic}
%If $\mu(\mathbf{z}) \in L_{\rm valid}(\mathcal{M})$, then $\mu'(\mathbf{z}) = \mu(A^\dag \mathbf{z}) \in L_{\rm %valid}(\mathcal{M})$ for any linear symplectic matrix $A$ on $\mathcal{M}$.
%\end{lemma}

%\begin{proof}
%Because $A^\dag \Sigma A = \Sigma$, the transformed covariance matrix also satisfies the CUP.  The fact that the max-ent condition is preserved can be seen using the expression (\ref{eq:GaussianDistFromCovMatrix}) of a max-ent distribution as a Gaussian function.
%\end{proof}

\subsubsection{Perfect knowledge of quadrature variables}

For a single degree of freedom, the classical uncertainty principle states that there is a trade-off between the degree of certainty an agent can have about each of two variables in a canonically conjugate pair.  One form of this trade-off is to have perfect knowledge of one of the variables and no knowledge of the other.  Such states of knowledge form an interesting subset of the valid epistemic states which correspond in quantum theory to eigenstates of \emph{quadrature operators}, which are linear combinations of the position and momentum operator (strictly speaking, these eigenstates are not normalized vectors in the Hilbert space, but we will not concern ourselves with these mathematical subtleties).  Here, we show how such states of knowledge can be described and shown to be consistent with the epistemic constraint.

Consider a Gaussian distribution $\mu$ on a phase space $\mathcal{M}$ for a single degree of freedom, with mean position $a$, mean momentum $b$ and covariance matrix
\begin{equation}
  \gamma_s =
    \begin{pmatrix}
	2s^{2} & 0 \\
	0 & 2\lambdabar^2 s^{-2}
	\end{pmatrix}
  \,,
\end{equation}
where $s$ is a real parameter.  This covariance matrix clearly saturates the CUP for all $s$, and being Gaussian the distribution is therefore in $L_{\rm valid}(\mathcal{M})$.  The corresponding epistemic state of the form (\ref{eq:GaussianDistFromCovMatrix}) factorizes into a Gaussian distribution over $q$ and a Gaussian distribution over $p$, i.e.,
\begin{equation}
\label{eq:GaussGauss}
  \mu(q,p) =  G_{a,s}(q) G_{b,\lambdabar s^{-1}}(p)\,,
%  \mu_{(a,b),s}(q,p) = G_{a,s}(q) G_{b,\lambdabar s^{-1}}(p)\,,
\end{equation}
where $G_{a,s}$ is a single-variable Gaussian with mean $a$ and standard deviation $s$, i.e.,
\begin{align}
  G_{a,s}(q) &= \frac{1}{2\sqrt{\pi}s} \exp\bigl(-\frac{(q-a)^2}{4s^2}\bigr) \,, \\
  G_{b,\lambdabar s^{-1}}(p) &= \frac{1}{2\sqrt{\pi}\lambdabar s^{-1}} \exp\bigl(-\frac{(p-b)^2}{4\lambdabar^2 s^{-2}}\bigr) \,.
\end{align}
Now consider the limit $s\rightarrow 0$.   This is the limit where uncertainty about position vanishes and uncertainty about momentum diverges.  Given that decreasing $s$ corresponds to squeezing the epistemic state along the position axis, we can also consider the limit $s\rightarrow 0$ to be the limit of infinite squeezing.  In this limit, the position distribution $G_{a,s}(q)$ becomes a Dirac delta function $\delta(q-a)$ centred at $a$,
\begin{equation}
  G_{a,s}(q) \to \delta(q-a) \equiv \lim_{s\to 0} \frac{1}{2\sqrt{\pi}s} \exp\bigl(-\frac{(q-a)^2}{4s^2}\bigr) \,,
\end{equation}
and the momentum distribution $G_{b,\lambdabar s^{-1}}$ approaches a uniform distribution.  Thus, the epistemic state corresponding to infinite squeezing along position (with mean position $q=a$) is the limit $s \to 0$ of Eq.~\eqref{eq:GaussGauss},
\begin{equation}
%  \mu_{q=a}(q,p) = \lim_{s\to 0} \mu_{(a,b),s}(q,p) \propto \delta(q-a) \,,
  \mu_{q=a}(q,p) = \lim_{s\to 0} G_{a,s}(q) G_{b,\lambdabar s^{-1}}(p)  \propto \delta(q-a) \,.
\end{equation}
It is the analogue within ERL mechanics of the eigenstate of the position operator $\hat{q}$ with eigenvalue $a$.

By applying a rotation to the phase space, which is a linear symplectic transformation, we can obtain related distributions of the form
\begin{equation}
  \mu_{q_\theta = a_\theta}(q,p) \propto \delta(q_\theta - a_\theta) \,,
\end{equation}
for any positive $a_{\theta}$ where $q_{\theta} = \cos(\theta)q +\sin(\theta)p$ is an arbitrary \emph{quadrature} (the quadrature variables are the linear combinations of position and momentum). This distribution corresponds to having perfect knowledge of the quadrature $q_{\theta}$ and no knowledge of the canonically conjugate quadrature $q_{\theta+\pi}$.  Because every linear symplectic transformation takes a valid epistemic state to another valid epistemic state, it follows that all these distributions are valid.  They are the analogues within ERL mechanics of the eigenstates of the quadrature operators $\hat{q}_{\theta} = \cos(\theta)\hat{q} +\sin(\theta)\hat{p}$.  See Fig.~\ref{fig:Quadratures} for an illustration.

\begin{figure}
\begin{center}
\includegraphics[width=1\hsize]{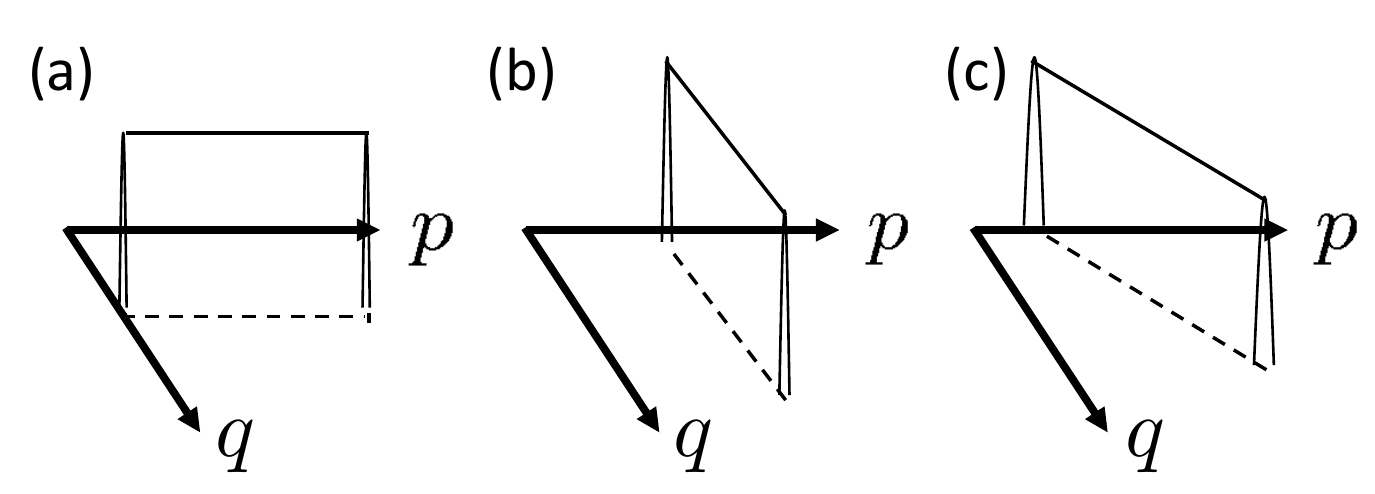}
\caption{Valid Liouville distributions corresponding to perfect knowledge of (a) position, (b) momentum, (c) a general quadrature.}\label{fig:Quadratures}
\end{center}
\end{figure}

For a composite system of $n$ canonical degrees of freedom, an argument paralleling the one above shows that one can have perfect knowledge of all of the canonical positions and no knowledge of any of the canonical momenta, corresponding to the distribution
\begin{equation}
  \mu_{\mathbf{q}=\mathbf{a}}(\mathbf{z}) = \prod_{i=1}^n \mu_{q_i=a_i}(q_i,p_i)\,,
\end{equation}
where $\mathbf{z} \in \mathcal{M}$.  Because one can implement a linear symplectic transformation on each component, it is clear that one can prepare any product of valid epistemic states for the components. That is, one can have perfect knowledge of some arbitrary quadrature for each component. But one can also have correlated epistemic states, as we now demonstrate.

\subsubsection{Correlated epistemic states}
\label{subsec:EPR1}

Given that linear symplectic transformations can mix the canonical variables of different systems, this allows states of perfect knowledge of relational and collective variables.  For instance, for a pair of systems $A$ and $B$ each with a single degree of freedom, the following linear map is easily shown to be symplectic:
\begin{align}
  q_A &\to q_A-q_B\,, &\quad p_A &\to p_A-p_B\,, \nonumber \\
  q_B &\to q_A+q_B\,, &\quad p_B &\to p_A+p_B\,.
\end{align}
Therefore, by starting with a valid epistemic state for which one has perfect knowledge of $q_A$ and $p_B$, we can map to a valid epistemic state with perfect knowledge of the relative position $q_A-q_B$ and the total momentum $p_A+p_B$ while having no knowledge of the canonically conjugate variables $q_A+q_B$ and $p_A-p_B$.

The particular epistemic state for which the relative position and total momentum are both known to vanish, $q_A-q_B=0$ and $p_A+p_B=0$, corresponds to the quantum state described by Einstein, Podolsky and Rosen (EPR)~\cite{EPR35}, which exhibits maximal entanglement between the pair of systems.  We define a valid epistemic state with these properties explicitly as the limit of Gaussians that are squeezed along $q_A-q_B$ and $p_A+p_B$,
\begin{multline}
  \mu^{\text{corr}}_{AB}(q_A,p_A,q_B,p_B) \\
   = \lim_{s\to 0} G_{0,s}(q_A-q_B) G_{0,\lambdabar s^{-1}}(p_A-p_B)
   \\ \times G_{0,s^{-1}}(q_A+q_B) G_{0,\lambdabar s}(p_A+p_B)\,.
   \label{eq:corr}
\end{multline}
Clearly, this distribution corresponds to knowing that $q_A-q_B=0$ and $p_A+p_B=0$,
\begin{equation}
  \mu^{\text{corr}}_{AB}(q_A,p_A,q_B,p_B) \propto \delta(q_A-q_B)\delta(p_A+p_B)\,.
  \label{eq:corr2}
\end{equation}
So we see that maximal bipartite entanglement in ERL mechanics is modelled by an epistemic state that describes perfect correlations between the pair of systems.

Such an epistemic state is easily generalized to the case of two copies of any system.  If the system has a $2n$-dimensional phase space, and the coordinates of each subsystem are paired into a set $\{q_{iA},p_{iA},q_{iB},p_{iB};i=1,\ldots,n\}$, then the analogue of the EPR state is
\begin{equation}
  \mu^{\text{corr}}_{AB}(\mathbf{z}_A,\mathbf{z}_B) =
  %\frac{1}{\Omega_{\mathcal{M}}}
 \prod_{i=1}^n \mu^{\text{corr}}_{AB}(q_{iA},p_{iA},q_{iB},p_{iB})\,.
 \label{eq:corr3}
\end{equation}
This epistemic state will be useful in our development of the formalism of ERL mechanics in Section~\ref{sec:Characterisation}.

\subsubsection{Variables whose values can be jointly known}

Given that linear symplectic transformations keep us within the space of valid epistemic states, one can have perfect knowledge of any set of variables that are the image of the set of canonical positions under a linear symplectic transformation.  To characterize these, recall that the Poisson bracket between two functions $f$ and $g$ of the canonical coordinates is
\be
\{f,g\}_{\rm PB}= \sum_i \left( \frac{\partial f}{\partial q_i} \frac{\partial g}{\partial p_i}-\frac{\partial f}{\partial p_i} \frac{\partial g}{\partial q_i} \right).
\ee
The set of canonical positions clearly all commute relative to the Poisson bracket, and no canonical momentum can be added to this set while maintaining commutativity. Furthermore, the Poisson bracket is preserved by symplectic transformations
%We therefore can conclude that in ERL mechanics one have perfect knowledge of a set of quadrature variables if and only if
 and therefore the sets of variables for which one can have perfect knowledge in ERL mechanics are precisely the sets of quadrature variables that commute relative to the Poisson bracket.  This is the analogue in ERL mechanics of the fact that in quantum theory one can jointly measure a set of observables if and only if they are commuting relative to the matrix commutator.

Whereas in quantum mechanics, commutation relative to the matrix commutator is a criterion for two observables to be jointly \emph{measurable}, in ERL mechanics, commutation relative to the Poisson bracket is a criterion for two variables to be jointly \emph{known}.
%It follows that this criterion can be used to characterize the valid epistemic states.

\subsubsection{The impossibility of concentrating uncertainty in a subsystem}

Consider the case of a composite system.  Although the epistemic restriction constrains what can be known about the ontic state of the whole system, it is not immediately obvious whether it also constrains what can be known about the ontic state of each subsystem.
For instance, given that it is possible to decrease one's uncertainty about one canonical variable by increasing it for its canonically conjugate partner, might it also be possible to decrease one's uncertainty about a pair of canonically conjugate variables by increasing it for a different pair of canonically conjugate variables?
%For instance, might it be possible to concentrate one's uncertainty into one subsystem, such that one can have perfect knowledge of the other?
As it turns out, this is not possible, as we now show.

Consider a system consisting of two subsystems $A$ and $B$. Let $\mu_{AB}$ be an epistemic state for the joint system that satisfies the epistemic restriction.  It is then straightforward to prove that the marginals of $\mu_{AB}$ on system $A$ and $B$, denoted $\mu_A$ and $\mu_B$ respectively, will also satisfy the epistemic restriction.  First, note that if $\gamma_{AB}, \gamma_A$ and $\gamma_B$ denote the covariance matrices of $\mu_{AB}$, $\mu_{A}$ and $\mu_{B}$ respectively, then
\begin{equation}
\gamma_{AB} =
\begin{pmatrix}
\gamma_A & X \\
X^{\dag } & \gamma_B
\end{pmatrix}
\,,  \label{gammablock}
\end{equation}
for some matrix $X$.  Note also that $\Sigma_{AB}$ has the form
\begin{equation}
\Sigma_{AB} =
\begin{pmatrix}
\Sigma_A & 0 \\
0 & \Sigma_B
\end{pmatrix}
\,,  \label{sigmablock}
\end{equation}
where $\Sigma_{AB}$, $\Sigma_A$ and $\Sigma_B$ are defined as in Eq.~\eqref{eq:SymplecticForm} for the phase spaces of the composite $AB$, the subsystem $A$ and the subsystem $B$ respectively.

From the fact that $\gamma_{AB} + i\lambdabar \Sigma_{AB} \geq 0$, we can infer that $\gamma_A + i\lambdabar \Sigma_A \geq 0$ and $\gamma_B + i\lambdabar \Sigma_B \geq 0$ using the following well-known result from linear algebra (\cite{HorneJohnson}, p.~472): for real matrices $a,b,c$ and $d$,
%A standard theorem of linear algebra (\cite{HorneJohnson}, p.~472) asserts that a partitioned matrix $\left(
%\begin{smallmatrix}
%a & b \\
%b^{T} & c
%\end{smallmatrix}
%\right)$ is positive definite if and only if $a\ge 0$ and $c\geq b^{T}a^{-1}b$, which applied
%Using the identity~\cite{HorneJohnson}
\begin{equation} \label{eq:partitionedmatrix}
\begin{pmatrix}
a & b \\
b^{T} & c
\end{pmatrix}
> 0 \quad \text{iff }\quad a > 0 \text{ and } c > b^{T}a^{-1}b \,.
\end{equation}
(This result will be used on many occasions in this article. Note that the conditions of positive definiteness can be replaced by conditions of positive semi-definiteness by continuity for the covariance matrices of Gaussian states that we consider here.)

Finally, because the marginal of a Gaussian distribution is also a Gaussian, if $\mu_{AB}$ satisfies the max-ent condition, then $\mu_A$ and $\mu_B$ do as well.

\begin{figure}
\begin{center}
\includegraphics[width=1\hsize]{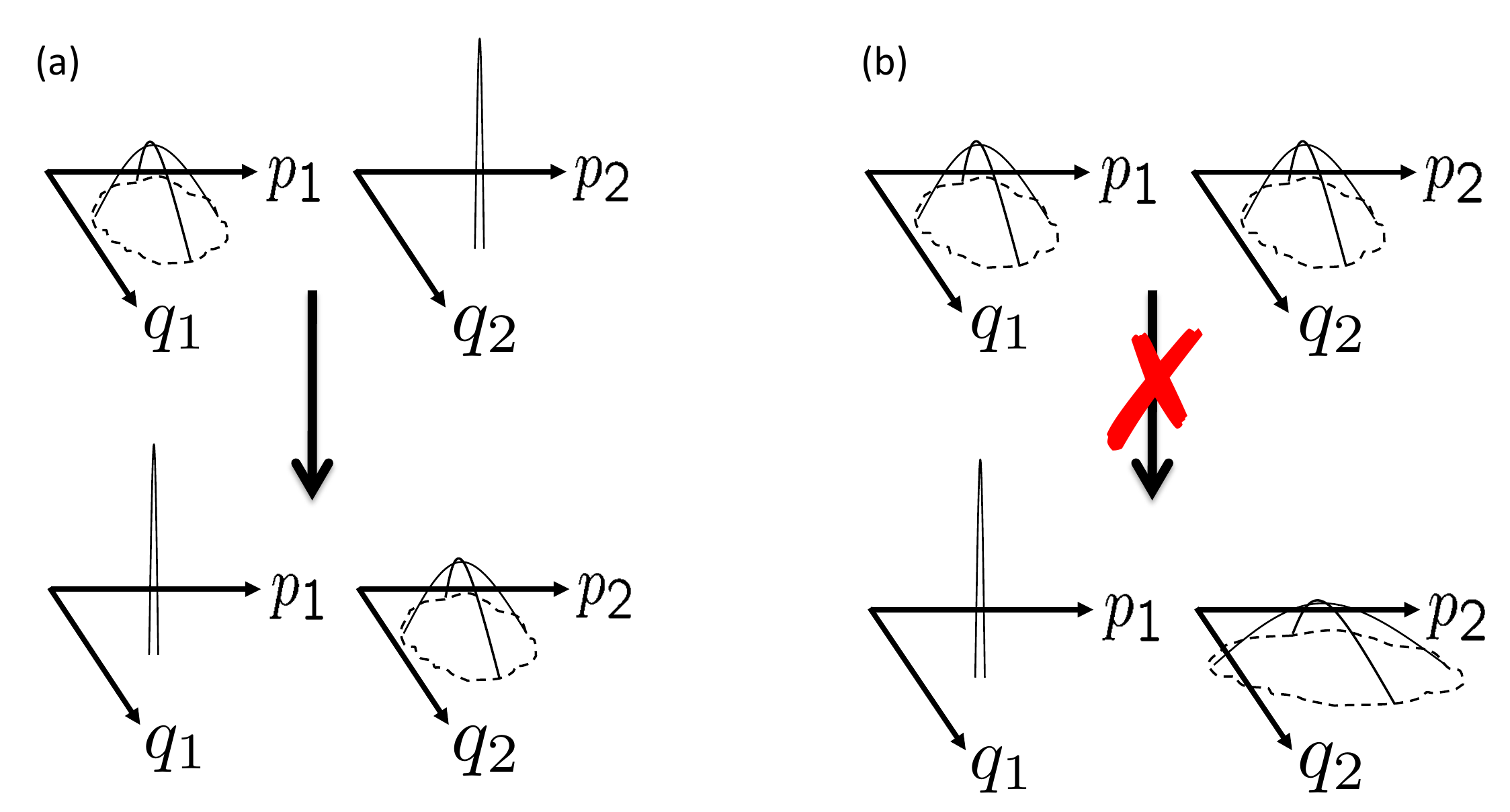}
\caption{A transformation that shuffles uncertainty from one subsystem to another, as illustrated in (a), is allowed within ERL mechanics.  However, a transformation as illustrated in (b) that concentrate uncertainty into one subsystem, such that one can have perfect knowledge of the other, is not allowed within the theory.}\label{fig:Gromov}
\end{center}
\end{figure}

This result can easily be generalized to any partition of the symplectic vector space of the whole system into symplectic subspaces, and therefore applies to the marginals on \emph{virtual} as well as physical subsystems.  For instance, one cannot achieve certainty about the canonically conjugate pair of variables $(q_A - q_B,p_A - p_B$) by concentrating one's uncertainty into the canonically conjugate pair $(q_A + q_B, p_A + p_B)$.
Note that this result is a special case, applicable only to symplectic vector spaces, of a general result due to Gromov~\cite{Gro85}.  For an analysis of the consequences of Gromov's theorem for uncertainty in Liouville mechanics, see Refs.~\cite{GossonTextbook,Hsiao07,Gosson09}.

To summarize, Liouville's theorem only predicts that one's uncertainty about an isolated system cannot be reduced by a symplectic transformation.  Therefore, it leaves open the possibility that this uncertainty can be partitioned amongst interacting subsystems in such a way that one is left with no uncertainty about one of them.  However, this possibility is precluded by the result just described, which may be considered as a strengthening of Liouville's theorem.  If an observer had access to even a \emph{single} system with a phase-space distribution that violated the CUP, for instance a Gaussian state that saturates the CUP for a value  $\lambdabar'<\lambdabar$, then other systems which initially satisfied the CUP could be made to violate it by transferring the certainty from the special system to the system of interest.  However, whatever is the minimum uncertainty for any subsystem initially, $\lambdabar'$ in our example, becomes a lower bound on the uncertainty for any subsystem finally.  Uncertainty can never be ``squeezed out'' of one subsystem and into another; see Figure~\ref{fig:Gromov}.
This result is quite reminiscent of Bohr's defense of the consistency of the uncertainty principle \cite{BohrReplytoEPR}, wherein he appealed to the unavoidable uncertainty in the initial position and momentum of the \emph{apparatus} to show that one could not reduce one's uncertainty about the position and momentum of the system\footnote{Bohr presented this defense of the uncertainty principle as part of his reply to the paper by Einstein, Podolsky and Rosen \cite{EPR35}. As we'll show in Sec.~\ref{sec:EPR}, ERL mechanics can reproduce the correlations in the original EPR thought experiment and indeed delivers the sort of interpretation of the correlations that EPR favoured, namely, one wherein position and momentum are jointly well-defined but not jointly known.  Even though Bohr sought to dispute this sort of interpretation in his reply, his description of the thought experiment makes explicit reference to the positions and momenta of the systems: ``In fact, even if we knew the position of the diaphragm relative to the space frame before the first measurement of its momentum, and even though its position after the last measurement can be accurately fixed, we lose, on account of the uncontrollable displacement of the diaphragm during each collision process with the test bodies, the knowledge of its position when the particle passed through the slit.''  Indeed, his argument for the consistency of the uncertainty principle makes no reference to the quantum formalism at all.  It reads better as an argument for the consistency of the uncertainty principle within ERL mechanics.  Nonetheless, Bohr denies the interpretation suggested by ERL mechanics: ``we have in each experimental arrangement suited for the study of proper quantum phenomena not merely to do with an ignorance of the value of certain physical quantities, but with the impossibility of defining these quantities in an unambiguous way.''  The only way we see to reconcile this tension in Bohr's reply is that Bohr believed that two quantities can be jointly \emph{well-defined} only if they can be jointly \emph{measured}.  In essence, Bohr was a radical positivist. Otherwise, why from the impossibility of two quantities being  jointly measured would he infer the impossibility of their being jointly well-defined, as opposed to merely inferring the impossibility of their being jointly known?}.

\subsubsection{Measurements of canonical variables}
\label{sec:SharpMmts}

We now consider what the epistemic restriction says about which measurements can be performed.  In particular, we consider what sorts of canonical variables can be jointly measured\footnote{Our development follows Ref.~\cite{SchreiberSpekunpublished}.}.
%To solve this problem, we will need to consider the scenario of the EPR experiment.

%\textbf{The predictive explanation}
%The typical sort of measurement one considers in Liouville mechanics corresponds to revealing the value of some function over the phase space $f(\mathbf{z})$ or some set of functions.  To understand which sets of functions correspond to valid measurements, we need to proceed with some care.
%We will take the epistemic restriction to imply a constraint on predictions, that is, inferences from events at one time to events at a later time.
%The problem is that the specification of which function one is measuring does not specify how the epistemic state for the system ought to be updated.
%or instance, if one considered a single system that was uncorrelated with all others, one would have no reason to rule out the possibility of a joint measurement of position and momentum because it could be that after the measurement was complete, these were completely randomized in such a way that the measurement is useless for predictions.  Such a measurement would only be useful for retrodiction.

At first glance, it might seem that the epistemic restriction could constrain the variables that can be jointly measured because the latter specify how the epistemic state of the system ought to be updated as a result of the measurement.  However, specifying which variable is measured only specifies what one can \emph{retrodict} about the ontic state of the system prior to the measurement, and says nothing about what one can predict about the ontic state of the system after the measurement.  This is analogous to how, in quantum theory, the observable that is measured specifies what one can retrodict about the system prior to the measurement and says nothing about how the system updates.  The fact that physicists have focussed upon the von Neumann-L\"{u}ders rule (i.e., the projection postulate) might generate the mistaken impression that the observable being measured \emph{does} fix how the quantum state updates.  However, there are many other update rules consistent with a given observable. For instance, a ``measure-and-reprepare'' update rule is one wherein regardless of the observable being measured and regardless of the measurement's outcome, some fixed quantum state is prepared after the measurement.  Similarly, in the case of Liouville mechanics, there are many possible rules for updating the epistemic state for any given set of variables being measured, in particular, a rule that prepares a fixed epistemic state after the measurement. Because of the possibility of such an update rule, the epistemic state after the measurement can always be made to satisfy the epistemic constraint \emph{regardless} of the variables being measured.

Nonetheless, it is possible to constrain the canonical variables one can jointly measure by considering measurements on one element of a pair of perfectly correlated systems as described in Sec.~\ref{subsec:EPR1}.  By assumption, the ontic dynamics is classical.  Therefore, if a measurement is made on particle $A$ and nothing is done to particle $B$, the ontic state of $B$ will not change.  However, the epistemic state for $B$ may well change as a result of learning about the ontic state of $A$ (via the measurement's outcome) and knowing that $A$ and $B$ are perfectly correlated.

For instance, if one could jointly measure the values of position and momentum on system $A$, then, by virtue of the fact that $A$ and $B$ are perfectly correlated in position and perfectly anticorrelated in momentum, one could infer the values of position and momentum on system $B$.  But this would correspond to having an epistemic state on $B$ that violates the epistemic restriction, so a joint measurement of position and momentum must be ruled out.

More generally, measuring the values of a set of canonical variables on $A$ implies learning the values of the same variables (modulo an inversion of the momenta) on $B$.  We have already seen that the only set of canonical variables that can be jointly \emph{known} according to the epistemic restriction are those that commute relative to the Poisson bracket.  It is also clear that if a set of variables on $B$ is obtained by taking the momentum inversion of a set of variables on $A$, then the first set commutes relative to the Poisson bracket if and only if the second does.  We therefore conclude that the only set of canonical variables that can be jointly \emph{measured} according to the epistemic restriction are those that commute relative to the Poisson bracket.

We have here considered only measurements of canonical variables. No other variables (for instance, nonlinear combinations of canonical variables) can be measured in ERL mechanics, a claim that we will justify in Sec.~\ref{sec:Characterisation}.  Furthermore, we have here only considered measurements wherein the outcomes are determined uniquely by the phase-space point.  The latter are the analogues of projective measurements in quantum theory.  The more general kind of measurement, for which the outcome is only determined probabilistically by the phase-space point, is the analogue of a \emph{positive operator-valued measure} in quantum theory.  We also leave the characterization of these to Sec.~\ref{sec:Characterisation}.
%We call these \emph{unsharp} and leave our discussion of them until Sec.~\ref{sec:Characterisation}.

\subsubsection{Transformations induced by measurements of canonical variables}

Finally, we must consider what the epistemic restriction says about how the epistemic state associated with a system is updated when that system is subjected to a measurement.  More precisely, given the information that the measurement has revealed some outcome, and given the epistemic state describing one's knowledge of what the ontic state of the system was at time $t$, prior to the measurement, what is the epistemic state describing one's knowledge of what the ontic state of the system is at time $t'$, after the measurement?\footnote{We have been careful in our description of the problem so as not to confuse two times in the problem: the time at which the agent assigns a given epistemic state and the time to which the agent's knowledge pertains.}

Consider the case of a quadrature measurement that is perfectly reproducible, in the sense that if the measurement is repeated in a sequence, with vanishing time between the measurements (and hence trivial evolution), the same outcome is always found.  Such measurements are perfectly consistent with the epistemic restriction, and thus are allowed in ERL mechanics.  What we will show now is that a measurement of one quadrature necessarily induces a completely unknown shift in the canonically conjugate quadrature.  That is, we will show the necessity of a \emph{disturbance} on the system as a result of measurement.

For simplicity, let this reproducible measurement be a measurement of position on a single system.  Suppose that one has perfect knowledge of the momentum of the system at time $t$, prior to the measurement.  We start by showing that if there were no disturbance to the momentum of the system as a result of the measurement of position, then after the measurement was complete, one would know both the position and the momentum of the system.  Note first that because the measurement is assumed to be reproducible, the distribution over position at time $t'$, after the measurement, must be a delta-function centered at the position revealed by the measurement.  Otherwise, there would be some probability of finding a different position upon repeating the measurement, contrary to the hypothesis of reproducibility.  Consequently, the final (i.e., post-measurement) position of the system is known based on the outcome of the reproducible measurement.  But note also that under the hypothesis that the momentum is not disturbed by the position measurement, one would also know the final momentum of the system based on one's knowledge of its initial momentum and the knowledge that it hasn't changed. Thus, no disturbance would imply the possibility of simultaneous knowledge of position and momentum and given that such knowledge is forbidden by the epistemic constraint it follows that if a measurement of position is to be possible, it cannot leave the momentum undisturbed.

Indeed, because the final position is perfectly known, the epistemic constraint dictates that the final momentum must be completely unknown.  But given the assumption that the \emph{initial} momentum is perfectly known, it follows that the position measurement must lead to a shift in momentum that is drawn uniformly at random from among all possible shifts.  The same argument could be run for the measurement of any quadrature, and so we have reached our desired conclusion: the only way to maintain the epistemic constraint is if a measurement of one quadrature necessarily induces a completely unknown shift in the value of the canonically conjugate quadrature\footnote{This feature of ERL mechanics is reminiscent of the language of ``uncontrollable disturbances'' used by Heisenberg and Bohr in their descriptions of quantum measurements.}.
%In particular, by Bohr in his reply to the EPR argument [cite] \footnote{It should be noted, however, that Bohr did not espouse a hidden variable approach, stating that ``we have in each experimental arrangement suited for the study of proper quantum phenomena not merely to do with an ignorance of the value of certain physical quantities, but with the impossibility of defining these quantities in an unambiguous way.''}

We have shown that if a measurement of a quadrature variable $q_{\theta}$ is performed in a reproducible manner, then the final epistemic state is the one wherein one has perfect knowledge of $q_{\theta}$ and no knowledge of $q_{\theta+\pi}$. This is analogous to how, in quantum theory, if a measurement of the quadrature operator $\hat{q}_{\theta}$ is performed in a reproducible manner (i.e., the state updates according to the projection postulate), then the final quantum state is an eigenstate of $\hat{q}_{\theta}$.  ERL mechanics provides a simple picture of the projection postulate applied to quadrature observables.  In this view, the collapse describes the change in an agent's knowledge of a system.  This change is not merely a Bayesian updating based on acquiring knowledge of the value of a quadrature, but a combination of such a Bayesian updating followed by a uniform probabilistic mixture of shifts in the canonically conjugate quadrature.

This unknown disturbance also explains how ERL mechanics reproduces the noncommutativity of measurements of canonically conjugate quadratures, that is, the fact that the statistics of outcomes of consecutive reproducible measurements of position and momentum depends on the order of the measurements.  Consider the quantum case first. Suppose the initial state is a position eigenstate. If the position measurement comes first, it has a deterministic outcome, while if it comes second, then it has a probabilistic outcome because the intervening momentum measurement collapses the quantum state to a momentum eigenstate.  Analogously, in ERL mechanics if the initial epistemic state is one wherein position is known perfectly, then if a position measurement comes first, it has a deterministic outcome, while if it comes second, then it has a probabilistic outcome because the intervening momentum measurement randomizes the position of the system.

\subsubsection{Modeling measurements with deterministic dynamics: the motility of the cut}
\label{sec:motilityofthecut}

The existence of an unknown disturbance might suggest that ERL mechanics presumes an underlying dynamics that is objectively stochastic.  This is not the case.
To see that the unknown disturbance is consistent with deterministic dynamics, it is sufficient to consider the measurement not as an external intervention but as a dynamical process (as was done for quantum mechanics by von Neumann~\cite{JvNtext}).  To do so, we imagine that the measurement couples the system to a probe through an interaction Hamiltonian $H= \chi q_A p_B$, where $q_A$ is the position of the system and $p_B$ is the momentum of the probe and $\chi$ is the interaction strength.  We imagine that the free Hamiltonian is negligible compared to the interaction Hamiltonian for the duration of the measurement. Recalling that momentum is the generator of translations, it follows that this interaction Hamiltonian causes the probe particle to be shifted by an amount that is proportional to the initial position of the system.  Thus, by measuring the shift in the position of the probe, one can infer the position of the system.  However, in order to be able to infer the exact value of the shift in position of the probe, it is necessary that the initial position of the probe be perfectly known.  Consequently, for the apparatus to achieve a measurement of position, it is necessary that the epistemic state describing the initial ontic state of the probe particle be one wherein there is perfect knowledge of position and complete uncertainty about momentum.

Now note that because position is the generator of shifts in momentum, the interaction Hamiltonian $H= \chi q_A p_B$ also causes the system to have its momentum shifted by an amount that is proportional to the initial momentum of the probe.  Effectively, while the probe acquires information about the position of the system, the system acquires information about the momentum of the probe.  This is an instance of the action-reaction principle of classical mechanics.  Given that the initial momentum of the probe is completely unknown (as highlighted above), it follows that the system suffers a shift in its momentum which is also completely unknown.

We conclude that the momentum disturbance in a position measurement is not a result of underlying objective stochasticity.  Rather, the final momentum of the system is uniquely determined by the initial momentum of the probe, but by virtue of the complete uncertainty about the latter, we are left with complete uncertainty about the former.  Effectively, under deterministic dynamics, our knowledge of the system's momentum is infected by our uncertainty about the probe's.

The other fact that is highlighted by this analysis is that the predictions of ERL mechanics are insensitive to the position of the `cut' between what is treated internally to the theory and what is treated externally.  This is the analogue of von Neumann's demonstration of the `motility of the cut' in quantum theory~\cite{JvNtext}.

\subsection{Some quantum phenomena reproduced in ERL mechanics}
\label{sec:phenomena}

We have seen how to understand some basic quantum phenomena by the lights of ERL mechanics, for instance, the collapse of the wavefunction and noncommutativity of conjugate measurements.  We proceed to consider a few more examples.  Note that when we say that we have reproduced a quantum phenomena, we are not claiming that we are necessarily reproducing all quantitative
predictions of quantum theory related to that phenomenon.  Rather, we are claiming that we are reproducing precisely those aspects of the phenomenon that have been hitherto considered to rule out any explanation of the phenomenon in terms of a classical worldview.

\subsubsection{The EPR thought experiment}
\label{sec:EPR}

Consider the thought experiment proposed by Einstein, Podolsky and Rosen~\cite{EPR35}.  A pair of particles, denoted $A$ and $B$, are prepared in a quantum state such that they are correlated in their position along some axis $\hat{x}$. The correlation is described by the EPR entangled state $|\Psi^{\textrm{corr}}\rangle = \int \textrm{d}q_A\,\textrm{d}q_B\,\delta(q_A-q_B)|q_A\rangle|q_B\rangle$.  The pair of particles are distributed to two points that are spatially separated (along an axis orthogonal to $\hat{x}$).
%A pair of particles are prepared in a maximally entangled state having a wavefunction $\Psi(x_1,x_2) = \delta(x_1 - %x_2)$, and distributed to two distant points.
%\begin{equation}
%|\Psi\rangle = \int \textrm{d}x_1 \textrm{d}x_2 \delta(x_1 - x_2) |x_1 \rangle \otimes |x_2\rangle
%\end{equation}
If a measurement of position is implemented on $A$, then the quantum formalism states that one can immediately predict with certainty what would be the outcome of a measurement of position on particle $B$.  Similarly, if a measurement of momentum is implemented on particle $A$, then one can predict with certainty what would be the outcome of a measurement of momentum on particle $B$.

EPR point out that if the wavefunction is taken to be a complete description of reality, then a free choice made in one region of space will instantaneously effect the ontic state in another distant region of space.  Specifically, the free choice of measurement at particle $A$ determines whether the wavefunction of particle $B$ becomes an eigenstate of position or an eigenstate of momentum, and given that these describe different ontic states under the assumption that the wavefunction is a complete description of reality, it follows that the free choice effects the distant reality.  EPR took this to be in conflict with the principle of relativity.  On the other hand, EPR argued, if the wavefunction merely described one's knowledge of an underlying reality, then the experiment needn't be in conflict with relativity.  As an observer learns the outcome of the measurement on particle $A$,  they merely update their knowledge of the ontic state of particle $B$. ERL mechanics is precisely the sort of hidden variable model that allows the EPR experiment to be explained in this sort of way, as we now demonstrate.

ERL mechanics models the EPR entangled state by a distribution over the two-particle phase space of the form $\mu^{\text{corr}}_{AB}(q_A,p_A,q_B,p_B)\propto \delta(q_A - q_B)\delta(p_A+p_B)$, describing perfect correlation of the particles' positions and perfect anti-correlation of the particles' momenta.  This epistemic state was highlighted in Sec.~\ref{subsec:EPR1}.  The marginal on the phase space of either particle is the completely uniform distribution.  Consequently, prior to learning the outcome of the measurement on particle $A$, an observer knows nothing of the position or the momentum of particle $B$.  If she measures the position of particle $A$, then by virtue of knowing that the positions of the two particles are initially perfectly correlated, she can infer the position of particle $B$. On the other hand, if she measures the \emph{momentum} of particle $A$, then by virtue of knowing that the momenta of the two particles are initially perfectly anti-correlated, she can infer the momentum of particle $B$. In both cases, particle $B$ has some definite position and momentum all along that does not change in any way as a consequence of her measurement on particle $A$.  All that changes as a result of this measurement is how the observer refines her knowledge of the ontic state of particle $B$.  She either refines her knowledge of its position or she refines her knowledge of its momentum.  No ``spooky action at a distance'' is required to understand the EPR experiment if one adopts the interpretation offered by ERL mechanics.

We emphasize that we are not arguing that a $\psi$-epistemic local hidden variable model could explain \emph{all} quantum correlations, only that the particular correlations described in the EPR experiment can be so explained (in precisely the way that EPR suggested they should). This is not at odds with Bell's theorem because the correlations in the EPR experiment do not violate a Bell inequality.  Of course, because it is locally causal by construction, ERL mechanics cannot hope to reproduce Bell-inequality violations.  Such violations are one of the quantum phenomena that ERL mechanics emphatically \emph{cannot} reproduce, not even qualitatively.
The fact that it is possible to find a local hidden variable model for the original EPR set-up with measurements restricted to quadrature observables, is well-known~\cite{BellEPR,OuCVBell,Reid09}.

\subsubsection{The no-cloning theorem}

%RWS: I think I've finally properly understood the role of Hamiltonian dynamics for the question of cloning.  First of all, to prove a no-cloning theorem, one needn't make use of the constraint of Hamiltonian dynamics.  The data processing inequality is already enough to see that an unknown distribution cannot be cloned.  If it were, one would have increased the amount of information encoded in the pair of systems.  Where Hamiltonian dynamics becomes relevant is in proving a no-*broadcasting* theorem. Broadcasting is what seems to be possible classically.  Simply imagine a physical process which measures Q and P of the system, then prepares another system in that same ontic state, then forgets what the outcome of the measurement was. Alternatively, simply imagine that the system undergoes dynamics Q_s'=Q_s,P_s'=P_s, while the ancilla undergoes dynamics Q_a'=Q_s, P_a'=P_a.  The initial state of the anciall, Q_a, P_a, is simply lost.  This dynamics is not Hamiltonian, because {Q_a',P_a'} \ne {Q_s,P_a}, so that the Poisson bracket is not preserved.
%RWS: So I'm thinking that a better way to do this section might be to as why we can't broadcast a given extermal Liouville distribution (thereby achieving something that would have no counterpart in the quantum formalism given that any joint state with extremal reductions is a product state.

Imagine one is given a system prepared in an unknown quantum state $|\psi\rangle$.  A cloning process is one which adjoins to the system an ancilla in a fiducial state $|\chi\rangle$ and implements the map $|\psi\rangle |\chi\rangle \to |\psi\rangle |\psi\rangle $ for all $|\psi\rangle \in \mathcal{H}$.  Given that $|\psi\rangle$ is unknown, the map cannot depend on $|\psi\rangle$.  No such process exists in quantum theory\cite{NoCloning}.  In fact, one cannot even clone a pair of non-orthogonal states.  That is, if the system is prepared in a quantum state drawn from the set $\{ |\psi_1\rangle,|\psi_2\rangle \}$, where $0< |\langle \psi_1|\psi_2 \rangle| < 1$, then no process can implement the map
\be
|\psi_k\rangle |\chi\rangle \to |\psi_k\rangle |\psi_k\rangle \textrm{ for } k\in\{1,2\}.
\ee
The proof is as follows.  Unitary dynamics preserves inner products, so if a process is to be implemented by a unitary, it must preserve inner products.  In the cloning process, the magnitude of the inner product between the two possible initial states is
\be
|(\langle \psi_1 |\langle \chi|)(|\psi_2\rangle |\chi\rangle)|=|\langle \psi_1 |\psi_2\rangle|
\ee
while the inner product between the two possible final states is
\be
|(\langle \psi_1 |\langle \psi_1|)(|\psi_2\rangle |\psi_2\rangle)|=|\langle \psi_1 |\psi_2\rangle|^2.
\ee
Thus, the magnitude of the inner product is preserved only if $|\langle \psi_1 |\psi_2\rangle|=0$ or $1$, which implies that the two states are orthogonal or collinear.  Irreversible quantum operations will not help because they necessarily lead to an \emph{increase} in the quantum fidelity, while a cloning process requires a decrease in this fidelity.

If one takes an ontic view of quantum states, then given that classically the ontic state of a system can always be measured and then copied, it would appear that no-cloning is a nonclassical phenomenon. By contrast, if one adopts an epistemic view of quantum states, then the cloning process is properly understood as a process which clones the applicability of a given state of knowledge and it is seen to occur even classically.  We present the analogue of the no-cloning of two non-orthogonal quantum states.  Suppose that we are told that a system has been prepared by sampling its ontic state $\mathbf{z}$ from the distribution $\mu_1(\mathbf{z})$ or from the distribution $\mu_2(\mathbf{z})$.  Suppose that $\mu_1(\mathbf{z})$ and $\mu_2(\mathbf{z})$ are nondisjoint, that is, they overlap in some part of the phase space, $\mu_{1}(\mathbf{z})\mu_{2}(\mathbf{z}) \ne 0$ for some $\mathbf{z}$.  A cloning process is one which adjoins to the system an ancilla (with the same ontic state space as the system) prepared in a fiducial epistemic state $\nu(\mathbf{z}')$ and transforms the ontic state in such a way that the following map over epistemic states is induced:
\be
\mu_k(\mathbf{z})\nu(\mathbf{z}') \to \mu_k(\mathbf{z})\mu_k(\mathbf{z}') \textrm{ for } k\in\{1,2\}.
\ee
To see that this cloning process is impossible, we first define the \emph{classical fidelity} between distributions $\mu_1(\mathbf{z})$ and $\mu_2(\mathbf{z})$ as $\int \textrm{d}\mathbf{z} \sqrt{\mu_1(\mathbf{z})}\sqrt{\mu_2(\mathbf{z})}$; it is 0 if the distributions are disjoint and 1 if they are identical (the analogy to the magnitude of the inner product between quantum states should be clear).  It then suffices to note that a pair of distributions can encode a bit of information and that the classical fidelity between the distributions is a measure of their indistinguishability.  Given that the amount of information in an encoding cannot be increased by processing (this is the content of the data processing inequality), it should not be possible to increase their distinguishability, i.e. it should not be possible to decrease the fidelity by any processing.  At best, one can preserve it.  Let us consider what this constraint implies.
%according to Liouville's theorem \cite{Arn97} the classical fidelity between probability distributions on the phase %space is preserved under Hamiltonian dynamics while it decreases in the cloning process.
The classical fidelity between the two possible initial states is
\begin{multline}
  \int d\mathbf{z}\,d\mathbf{z}'\,\sqrt{\mu_{1}(\mathbf{z})\nu(\mathbf{z}')}
  \sqrt{\mu_{2}(\mathbf{z})\nu(\mathbf{z}')} \\
  =\int d\mathbf{z}\,\sqrt{\mu_{1}(\mathbf{z})}\sqrt{\mu_{2}(\mathbf{z})}\,,
\end{multline}
where we have used the fact that $\int d\mathbf{z}'\,\nu(\mathbf{z}')=1$, whereas between the two possible
final epistemic states it is
\begin{multline}
  \int d\mathbf{z}\,d\mathbf{z}'\,\sqrt{\mu_{1}(\mathbf{z})\mu_{1}(\mathbf{z}')}\sqrt{\mu_{2}(\mathbf{z})\mu_{2}(\mathbf{z}')}  \\
  =\left( \int d\mathbf{z}\,\sqrt{\mu_{1}(\mathbf{z})}\sqrt{\mu_{2}(\mathbf{z})}\right)^{2}\,.
\end{multline}
For the classical fidelity to be preserved, we require
\begin{equation}
  \int d\mathbf{z}\,\sqrt{\mu_{1}(\mathbf{z})}\sqrt{\mu_{2}(\mathbf{z})}
  =\left( \int d\mathbf{z}\,\sqrt{\mu_{1}(\mathbf{z})}\sqrt{\mu_{2}(\mathbf{z})}\right)^{2},
\end{equation}
which implies that $\int d\mathbf{z}\,\sqrt{\mu_{1}(\mathbf{z})}\sqrt{\mu_{2}
(\mathbf{z})}=0$ or $1$, or equivalently, that $\mu_{1}(\mathbf{z})\mu_{2}
(\mathbf{z})=0$ or $\mu_{1}(\mathbf{z})=\mu_{2}(\mathbf{z})$.  Thus, a pair of
epistemic states can be cloned if and only they are
disjoint or identical.
%Irreversible maps, such as
%convex combinations of Hamiltonian maps, will not help because they lead to
%an \emph{increase} in the classical fidelity, whereas cloning requires a
%decrease.

Note that the proof proceeds in direct analogy with the quantum proof, where
the role of orthogonality and Hilbert space inner
product are played by disjointness and classical
fidelity respectively. That there is a no-cloning theorem for non-disjoint probability distributions has also been noted in Refs.~\cite{DPP02,CF96,Har99,Spe07}  \footnote{A similar point can be made about the phenomenon of quantum chaos. Many researchers have been puzzled by the apparent differences between the classical and quantum signatures of chaos. While classical states of an isolated system can exponentially diverge under Hamiltonian chaotic
evolution, quantum states of an isolated system cannot because the inner
product between two quantum states is invariant under unitary evolution.
However, the analogy between quantum states and Liouville distributions
suggests that the quantum inner product should not be compared with the
distance in phase space but rather with the overlap of the Liouville
densities. One can then reconcile the signatures of classical and quantum
chaos \cite{Bal94, Eme01, EmersonBallentine01}.}.

%In quantum theory, there is a no-cloning theorem.   A simple proof proceeds as follows.  Imagine one is given a system prepared in an unknown quantum state drawn from the set $\{ |\psi_k\rangle \}$.  A cloning process is one which adjoins to the system an ancilla in a fiducial state $|\chi\rangle$ and implements the map $|\psi_k\rangle |\chi\rangle \to |\psi_k\rangle |\psi_k\rangle $ for all $k$. Given that $k$ is unknown, the map must be independent of $k$.  A particularly nice statement of the no-cloning theorem is that unless every pair of elements n the set $\{ |\psi_k\rangle \}$ are pairwise collinear or orthogonal,

We have yet to specify in which sense the epistemic restriction is necessary to properly model the quantum no-cloning theorem. We have seen that by simply defining cloning in terms of epistemic states rather than ontic states, one finds that certain pairs of epistemic states --nondisjoint ones-- cannot be cloned.  However, in Liouville mechanics (without the epistemic restriction) only \emph{mixed} epistemic states can be nondisjoint; the pure epistemic states are point distributions over the phase space and \emph{can} be cloned.  On the other hand, in ERL mechanics, the pure epistemic states (defined as those that are extremal in the convex set of epistemic states) are themselves states of incomplete knowledge and can be nondisjoint. It follows that only in ERL mechanics does one have an analogue of no-cloning for sets of pure quantum states.

There is one other sense in which the epistemic restriction is important for emulating all the limitations on cloning that are seen in quantum theory. Although in quantum theory it is impossible to have pure states as marginals without the state of the composite being a product state, one might wonder whether, in the context of ERL mechanics we could achieve a joint distribution over the composite system that has $\mu_{k}(\mathbf{z})$ as the marginal distribution for both subsystems but with the possibility of correlations between the systems.  We did not consider this possibility above, where we required the final distribution to be a product distribution.  Such a process would be a classical \emph{broadcasting} map.  As it turns out, the data processing inequality does not exclude this possibility.  Indeed, if it were not for the epistemic restriction, such a map could be realized.  One could measure the ontic state $\mathbf{z}$ of the system, prepare the target in the same ontic state, then forget the outcome of the measurement.  But such a measurement would violate the epistemic restriction, so this strategy will not work. Alternatively, one could simply implement the deterministic dynamics $(\mathbf{z},\mathbf{z'}) \to (\mathbf{z},\mathbf{z})$ on the pair of systems. This would achieve broadcasting regardless of the value of $\mathbf{z'}$, but it is not allowed because it does not preserve the Poisson bracket and hence is not Hamiltonian. Finally, there are Hamiltonian maps that can implement broadcasting for one particular value of $\mathbf{z'}$, but given the epistemic restriction, one cannot have such knowledge of $\mathbf{z'}$.

\subsubsection{Teleportation}
\label{sec:Teleportation}

We begin by providing the quantum description of teleportation for continuous variable systems.  The scenario is similar to that of the EPR experiment.  A pair of particles are prepared in the EPR entangled state $|\Psi^{\textrm{corr}}\rangle = \int \textrm{d}q_A\,\textrm{d}q_B\,\delta(q_A-q_B)|q_A\rangle|q_B\rangle$, that is, correlated in their position along the $\hat{x}$-axis, and distributed to Alice and Bob, who are presumed to be spatially separated (along an axis orthogonal to $\hat{x}$).  We assume trivial dynamics so that we can neglect dispersion over time.  A third party, Victor, prepares another particle, denoted $V$, in the quantum state $|\psi\rangle,$ (again, describing the position of the particle along the $\hat{x}$-axis) and passes it to Alice.  The identity of particle $V$'s quantum state is unknown to Alice and Bob. Their task is to implement a protocol that leaves particle $B$ in the quantum state $\left| \psi \right\rangle$.  The initial quantum state of the triple of particles is $|\psi\rangle |\Psi^{\textrm{corr}}\rangle$.
This initial state can be rewritten (preserving the order of the Hilbert spaces) as
\begin{equation}
  \frac{1}{2\pi\hbar} \int \textrm{d}a\,\textrm{d}b\,(D_{a,b} \otimes I)|\Psi^{\textrm{corr}}\rangle D^{\dag}_{a,b}|\psi\rangle\,,
\end{equation}
where $D_{a,b}=\exp(-\frac{i}{\hbar}(a\hat{p}-b\hat{q}))$ is the unitary operator that corresponds to a displacement in phase space by the vector $(a,b)$.

Note that the state of particles $V$ and $A$ appearing in the $a,b$ term in this decomposition is simply the joint eigenstate of the commuting pair of operators $\hat{q}_V-\hat{q}_A$ and $\hat{p}_V+\hat{p}_A$ associated with eigenvalues $a$ and $b$.  Consequently, if Alice measures $\hat{q}_V-\hat{q}_A$ and $\hat{p}_V+\hat{p}_A$ on the pair of particles in her possession, and obtains outcomes $a$ and $b$ respectively, then (assuming the projection postulate as the collapse rule) the total quantum state is updated to just one of the terms in the integrand.  Alice's two particles have been left in a maximally entangled state (a local phase space displacement of the EPR state), and Bob's particle has been left in the state $D^{\dag}_{a,b} |\psi \rangle$.  Therefore, to complete the protocol, Alice need only communicate $a,b$ to Bob, who then applies the unitary $D_{a,b}$ to his particle and leaves it in the state $\left| \psi \right\rangle$, as required.  The protocol succeeds regardless of the identity of $\left| \psi \right\rangle$, so Alice and Bob need not know its identity. Finally, note that if particle $V$ is entangled with another particle, say particle $C$, then the quantum state of the composite of particles $V$ and $C$ is transferred to the composite of particles $B$ and $C$, a phenomenon known as \emph{entanglement swapping}.

What is surprising about continuous variable teleportation, if one takes the view that quantum states are ontic, is that while it takes an infinite number of complex parameters to completely specify the quantum state, this state can be transferred from Alice to Bob by communicating only two real numbers.  Even if we restrict the unknown quantum state to be a Gaussian state, we still require five parameters to describe it (specifically, two for specifying the mean position and mean momentum and three for specifying the covariance matrix) but only two to transfer it.  On the other hand, if one takes the view that quantum states are epistemic, then teleportation is a protocol wherein someone's \emph{knowledge} about a system becomes applicable to another system and, as we shall see, such a transfer can be achieved with only two real parameters.

If Alice could jointly measure the position \emph{and} momentum of particle $V$, she could simply communicate this information to Bob who could then prepare particle $B$ with that precise position and momentum (this is essentially how teleportation is imagined to occur on Star Trek).  In this way, whatever Victor knew about particle $V$ would now pertain to particle $B$.   However, the epistemic constraint forbids such a joint measurement.  The magic of the teleportation protocol, by the lights of ERL mechanics, is that it provides a way of transferring the applicability of Victor's knowledge in spite of the epistemic constraint.

Teleportation of Gaussian states can be modeled in ERL mechanics as follows\footnote{The discussion provided here closely parallels the one provided in Ref.~\cite{Spe07}; see also~\cite{Cav04}.}.  The pair of particles shared by Alice and Bob are prepared in the epistemic state $\mu^{\text{corr}}_{AB}(q_A,p_A,q_B,p_B)\propto \delta(q_A-q_B)\delta(p_A+p_B)$ (the model of the EPR state), which corresponds to knowing the relative position of the two particles to be $q_B-q_A = 0$ and the total momentum to be $p_B+p_A =0$, which is to say that they are known to have the same position and opposite momenta.
%that particle $B$, in Bob's possession, has the same position and opposite momentum (along the $m$-axis) to particle $A$, in Alice's possession.
Alice makes a measurement on particles $V$ and $A$, both in her possession.  Specifically, she measures the relative position, $q_V-q_A$, and the sum of their momenta, $p_V+p_A$.  (This is allowed by the epistemic restriction because these variables have commuting Poisson bracket).   Combining this new data with her previous knowledge, Alice can infer what the relative position $q_V-q_B$ and relative momentum $p_V-p_B$ of particles $V$ and $B$ were prior to the measurement (because this is an inference based on pre and post selection of the triple of particles, the epistemic constraint need not apply, as discussed in Sec.~\ref{sec:EpistemicRestriction}).  Specifically, if Alice finds through her measurement that $q_V-q_A= a$ and $p_V+p_A=b$, then she infers that $q_V- q_B = (q_V-q_A) - (q_A-q_B) =a$ and that $p_V- p_B = (p_V + p_A) - (p_A+ p_B) =b$.
%Suppose she finds the relative position of particle $V$ to particle $B$ to be $a$ and their relative momenta to be $b$.
Given that the measurement is implemented on particles $V$ and $A$, it will not disturb the ontic state of particle $B$, so that the ontic state of particle $B$ after the measurement is precisely what it was prior to the measurement, namely, $q_B = q_V - a$ and $p_B = p_V - b$.  So Alice simply tells Bob to shift the position of particle $B$ by $a$ and its momentum by $b$, so that it will come to have the same position and momentum as particle $V$ had before the measurement.  In this way, whatever Victor knew about the ontic state of particle $V$ prior to the measurement now pertains to the ontic state of particle $B$ after the measurement.  Meanwhile, because particles $V$ and $A$ have undergone a measurement, there is an unknown disturbance to these, and consequently Victor's knowledge is no longer applicable to particle $V$ (which is why teleportation is not in conflict with no-cloning).  Had Victor initially known particle $V$ to have a particular correlation with particle $C$, then at the end of the protocol, he would judge particle $B$ to have this correlation with particle $C$, and so we also have a model of entanglement swapping for Gaussian states.  A formalized presentation of this entanglement swapping relation is given as Lemma~\ref{lemma:swap} in Sec.~\ref{sec:Characterisation}.

The reason that Alice can get away with communicating only two real parameters to Bob is that this amount of communication is sufficient (in the context of the teleportation protocol) for Bob to be able to prepare his particle in the ontic state that initially described the particle supplied by Victor. Once this is done, whatever knowledge Victor had of his particle's original ontic state, it now applies to Bob's particle, \emph{regardless of how many parameters are required to specify Victor's state of knowledge}.  Note furthermore that the transfer of the applicability of Victor's state of knowledge does not, strictly speaking, require \emph{any} communication from Alice to Bob. Suppose Alice only sends the outcome of her measurement to Victor, and not to Bob, so that Bob never does any correction operation on his particle.  Then, in the special case where Alice's measurement finds particles $V$ and $A$ to have had the same position and momentum, Victor can still conclude that whatever knowledge he initially had of his particle now pertains to Bob's particle.

\section{Operational Equivalence of ERL mechanics and Gaussian quantum mechanics}
\label{sec:Equivalence}

Having described some of the basic features of ERL mechanics, we will proceed to provide a complete operational formulation of the theory in Sec.~\ref{sec:Characterisation}.  We will then provide an operational formulation of a subtheory of quantum mechanics which we call \emph{Gaussian quantum mechanics} in Sec.~\ref{sec:GaussianQM}.  Finally, in Sec.~\ref{sec:proofmain}, we prove the main result of this article:
\begin{theorem}[Equivalence]
\label{thm:Gaussian}
Gaussian quantum mechanics is operationally equivalent to ERL mechanics with $\lambdabar=\hbar$.
\end{theorem}

%Having defined \emph{epistemically-restricted Liouville mechanics}, we now demonstrate that it is operationally equivalent to a subtheory of quantum mechanics which we call \emph{Gaussian quantum mechanics}.
A few definitions are required to make sense of this result.  An \emph{operational formulation} of a theory is one which only specifies what are the possible preparations, transformations and
measurements according to the theory, as well as a rule for computing the probability of the outcome of every measurement when performed on a given preparation followed by a given transformation.  An operational formulation of a theory needn't make any reference to ontological structure.
Two theories that are formulated operationally are said to be \emph{operationally equivalent}
%The notion of \emph{operational equivalence} is one that applies to a pair of theories and refers only to the operational predictions they make (and not to any ontological structure they might posit). Specifically, we say that two theories are operationally equivalent
 if there is a one-to-one mapping between the preparations, measurements and transformations that are allowed in the first theory and those that are allowed in the second, and if the statistics predicted for every possible experiment in the first theory are precisely the same as those predicted for the corresponding experiment in the second theory.  Finally, a subtheory of an operational theory is what one obtains by allowing only a \emph{subset} of the preparations, transformations and measurements that are allowed in the parent theory. \emph{Gaussian quantum mechanics} is the subtheory of quantum mechanics wherein the allowed preparations, measurements and transformations are those for which the associated Wigner representations are Gaussian functions.

\subsection{Operational formulation of ERL mechanics}
\label{sec:Characterisation}

The most general preparation in ERL mechanics is represented by a phase-space distribution.  We have already specified, in Sec.~\ref{sec:ERLTheory}, which distributions satisfy the epistemic restriction.  They are denoted $\mu \in L_{\text{valid}}(\mathcal{M})$ on a phase space $\mathcal{M}$.  Consequently, we have already specified the set of possible preparation procedures.  It therefore suffices to characterize the set of possible transformations and measurements.
%The only constraint these must satisfy is to preserve the epistemic restriction.

%We note that the existence of the maximally-correlated state of two identical systems $\mu^{\text{corr}}_{AB}$ within the theory, as defined in Sec.~\ref{subsec:EPR1}, will serve as a key mathematical tool in finding the allowed measurements and transformations.

\subsubsection{General measurements}
\label{sec:indicatorfns}

In Sec.~\ref{sec:SharpMmts}, we described which canonical variables could be measured jointly on a system.
%To begin, we consider how measurements are represented in Liouville mechanics (without an epistemic restriction). The standard sort of measurement in this context is one that is defined by a partitioning of $\mathcal{M}$ into regions and where the outcome simply specifies the region that contains the ontic state $\mathbf{z}$.
%The simplest way is to use the standard measurement theory of classical mechanics, wherein one requires complete information about the point in phase space. If the system is described by a Liouville distribution $\mu$, then the probability that the system is found to be within an area $d\mathbf{z}$ at $\mathbf{z}$ is $\mu(\mathbf{z})d\mathbf{z}$.
However, Liouville mechanics admits a more general form of measurement wherein the ontic state does not determine the outcome deterministically but only fixes the relative probabilities of various outcomes.  This occurs whenever the outcome of the measurement depends on other degrees of freedom besides the system of interest and the states of these are not completely known. For example, consider a system consisting of a single canonical degree of freedom. If it interacts with several ancillas via a quadratic Hamiltonian and measurements of quadrature variables are implemented upon the ancillas, the resulting measurement on the system will not in general yield full information about a singe quadrature, but rather will yield partial information about each of a pair of canonically conjugate quadratures.  As another example, if the system interacts with an auxiliary system that is subsequently ignored, the effective measurement on the system is not maximally informative (these sorts of measurements are in fact generic, because the ability to avoid all such noise is always an idealization within classical mechanics).

%A measurement is associated with a set of indicator functions, labeled by an element of the outcome space of the measurement.

The most general sort of measurement, which incorporates both the deterministic and probabilistic varieties, is associated with a set of \emph{indicator functions} on the phase space $\mathcal{M}$, that is, a set $\{ \xi_{\mathbf{y}}(\mathbf{z}) \}$ where $\xi_{\mathbf{y}}(\mathbf{z})\textrm{d}{\mathbf{y}}$ is the probability of obtaining a measurement outcome within $\textrm{d}\mathbf{y}$ of $\mathbf{y}$ given that the ontic state of the system is $\mathbf{z}$.  The variable $\mathbf{y}$ labels elements of the outcome space of he measurement. For instance, a measurement of position is associated with a set of indicator functions labeled by a position variable $q_0$, specifically, $\{ \xi_{\mathbf{q}}(\mathbf{z}) \propto \delta(q-q_0) \}$. For general measurements, the outcome space may be higher-dimensional.   Because $\xi_{\mathbf{y}}(\mathbf{z})$ is a probability density, we have $\xi_{\mathbf{y}}(\mathbf{z})\textrm{d}\mathbf{y} \ge 0$, and because \emph{some} outcome is certain to occur, we have $\int \textrm{d}\mathbf{y} \xi_{\mathbf{y}}(\mathbf{z}) =1$ for all $\mathbf{z}$.
Clearly, if the system is described by the epistemic state $\mu$ and a measurement described by the set of indicator functions $\{\xi_{\mathbf{y}} \}$ is performed, the probability density for outcome $\mathbf{y}$ is
%that an outcome within $\textrm{d}\mathbf{y}$ of $\mathbf{y}$ is obtained is given by
\begin{equation}
p(\mathbf{y}) = \int d\mathbf{z}\, \xi_{\mathbf{y}}(\mathbf{z}) \mu(\mathbf{z}) \,.
%p(\mathbf{y})\textrm{d}\mathbf{y} = \int d\mathbf{z}\, \xi_k(\mathbf{z}) \mu(\mathbf{z}) \,.
\end{equation}
%It is straightforward to verify that the positivity of the indicator functions $\xi_k$ guarantees that $p_k >0$, and that the condition $\sum_k \xi_k = 1$ ensures that $\sum_k p_k = 1$.

We now consider what constraints on the indicator functions follow from the epistemic restriction.

As discussed in Sec.~\ref{sec:SharpMmts}, the way to infer these constraints is by imagining the measurement to be performed on a system $A$ that is correlated with another system $B$, and applying the epistemic restriction to the final distribution assigned to $B$.
Specifically, we require that a valid indicator
function acting on a system $\mathcal{M}_A$ must always result in a valid
epistemic state on $\mathcal{M}_B$ when applied to any (possibly correlated) epistemic state on
$\mathcal{M}_A \times \mathcal{M}_B$.  This implies the following constraint.

\begin{proposition}[Valid indicator functions]
\label{thm:validindicatorfunctions}An indicator function on $\mathcal{M}$ is valid if and
only if, when normalized, it satisfies the epistemic constraint, that is,
\begin{equation} \label{eq:validindicatorfunctions}
\xi\ \text{\textrm{is valid iff}}\ \frac{\xi }{|\xi |}\in L_{\rm valid}(\mathcal{M})\,.
\end{equation}
\end{proposition}

\begin{proof}
First we prove necessity. If $A$ and $B$ are prepared in the perfectly-correlated state $\mu_{AB}^{\mathrm{corr}}$ and a measurement on $A$ yields the outcome associated with the indicator function $\xi_A$, then by Bayes' theorem, one updates the description of $AB$ to $\mu'_{AB} \propto \xi_A \mu_{AB}^{\mathrm{corr}}$.  Given that $\mu_{AB}^{\mathrm{corr}}$ describes perfect correlation for position and anti-correlation for momentum (see Eq.~\eqref{eq:corr2}), it follows that the marginal on $B$ is $\mu'_B(\textbf{z}) \propto \xi_A (\Lambda \textbf{z})$, where $\Lambda$ indicates momentum inversion, $\Lambda \equiv \textrm{diag}(1,-1,1,-1,\dots)$ (or equivalently, time inversion).
%which changes the sign of the momentum
%if $\mathcal{M}_A$ updates by the conditional $\xi_A$, then $\mathcal{M}_B$ is updated to $\xi_A/|\xi_A|$.
Thus, $\xi_A(\textbf{z})$ is a valid indicator function only if $ \xi_A(\Gamma \textbf{z})/| \xi_A(\Gamma \textbf{z})|$ is a valid state.  If a distribution is positive and satisfies the CUP, then so does its momentum inversion, consequently, it suffices to require that $ \xi_A(\textbf{z})/| \xi_A(\textbf{z})|$ is a valid state.

To prove sufficiency, we show that any Gaussian indicator function acting on system $A$ of any bipartite state $\mu_{AB} \in L_{\rm valid}(\mathcal{M}_A \times \mathcal{M}_B)$ yields a valid updated state on $B$.  We follow a similar proof to that found in Ref.~\cite{Eis02}.  Consider a measurement described by a Gaussian indicator function on $A$ with covariance matrix $\gamma'_A$, on a bipartite Gaussian state $\mu_{AB} \in L_{\rm valid}(\mathcal{M}_A \times \mathcal{M}_B)$ with covariance matrix $\gamma_{AB}$.  For clarity, we will assume these covariance matrices are both strictly positive-definite.  (A general proof for positive semi-definite matrices follows by appropriately using a pseudoinverse.) It is convenient to partition the matrix $\gamma_{AB}$ as
\begin{equation}
\gamma_{AB} =
\begin{pmatrix}
\gamma_{A} & X \\
X^{T} & \gamma_{B}
\end{pmatrix}
\, ,
\end{equation}
so as to respect the division of the joint state into subsystems $A$ and $B$.

The post measurement state $\mu'_B$ on system $B$ is found from the
the probability distribution $\mu_{AB}$ by conditionalizing on $\xi_A$ having been found on $A$ and marginalization on $A$ in the standard
manner.  We make use of the fact that Gaussian integrals performed over a subset of the variables concerned yields a Gaussian in terms of the remaining variables.  For our case, the relevant Gaussian distribution over $B$ has a covariance matrix given by the Schur complement~\cite{Eis02}
\begin{equation} \label{eq:gammaBprime}
  \gamma'_{B}=\gamma_{B}-X^{T}\left( \gamma'_A +\gamma_{A}\right)^{-1}X \,.
\end{equation}
We now need to confirm that $\gamma'_{B}$ satisfies the classical
uncertainty relation.  Given that $\gamma_{AB}$ and $\gamma'_{A}$ satisfy
the CUP, we have the relations
\begin{align}
  \gamma_{AB} + i\lambdabar\Sigma_{AB} &\geq 0\,, \\
  \gamma'_{A} - i\lambdabar\Sigma_A &\geq 0\,.
\end{align}
where in the second expression we have taken the complex conjugation of the usual expression.  Adding these two equations yields
\begin{equation}  \label{eq:CovMatrixAfterMeas}
\begin{pmatrix}
\gamma_{A}+\gamma'_{A} & X \\
X^{T} & \gamma_{B}+i\lambdabar\Sigma_B
\end{pmatrix}
\geq 0\,.
\end{equation}
Applying the result from linear algebra described in Eq.~\eqref{eq:partitionedmatrix}
%A standard theorem of linear algebra (\cite{HorneJohnson}, p.~472) asserts that a partitioned matrix $\left(
%\begin{smallmatrix}
%a & b \\
%b^{T} & c
%\end{smallmatrix}
%\right)$ is positive definite if and only if $a\ge 0$ and $c\geq b^{T}a^{-1}b$, which applied
and making use of Eq.~\eqref{eq:gammaBprime}, we find
\begin{equation}
  \gamma'_{B} + i \lambdabar \Sigma_B \geq 0 \,.
\end{equation}
\end{proof}

We note that, as a result of this theorem, the indicator function
\begin{equation}  \label{eq:PerfectlyCorrIndicatorFunction}
\xi^{\mathrm{corr}} \propto \mu^{\mathrm{corr}} \,,
\end{equation}
is a valid indicator function on $\mathcal{M} \times \mathcal{M}$.

%In this theory, it is clear that the resulting state on $\mathcal{M}_B$ obtained by performing a measurement on $\mathcal{M}_A$ is simply given by Bayesian updating. Note that, for a general correlated state on $\mathcal{M}_A \times \mathcal{M}_B$, it is possible to update the state for $\mathcal{M}_B$ in multiple different ways by choosing different measurements on $\mathcal{M}_A$. Schr{\"o}dinger referred to this effect in quantum mechanics as \emph{steering}. \sdb{Cite.}

The valid \emph{sets} of indicator functions $\{ \xi_{\mathbf{y}}(\mathbf{z}) \}$ are simply those consisting entirely of valid indicator functions such that $\int \textrm{d}\mathbf{y}\xi_{\mathbf{y}}(\mathbf{z})=1$ for all $\mathbf{z}$.  For example, if we take any valid indicator function with means at the origin of the phase space and consider the set obtained by acting all phase-space displacements on the latter, we obtain a valid set of indicator functions where the outcome of the measurement is labeled by a point in phase space.   We will denote elements of a general outcome space by $\mathbf{y}$.

\subsubsection{General transformations}

In Sec.~\ref{sec:Reversible}, we demonstrated that the valid reversible transformations within ERL mechanics were the linear symplectic transformations.  However, a general operational theory includes non-reversible transformations as well.  These can include dissipation due to coupling to another system (and subsequently marginalizing over that system), transformations due to a measurement being performed on the system, and irreversibility due to an agent lacking knowledge of which reversible transformation was implemented.  We now consider how such general transformations are described within ERL mechanics, and what constraints are forced upon these transformations by the epistemic restriction.

Recall that, by assumption, the dynamics is classical, but an observer might lack knowledge of the nature of the dynamics (for instance, if the environment with which the system is interacting is in an unknown physical state).  In this case, they assign a probability distribution over the possibilities for the dynamics. Such ignorance can always be characterized by a probability distribution over the final ontic states for every initial ontic state, that is, by a set of \emph{transition probabilities} $\eta(\mathbf{z}'_{A}|\mathbf{z}_A)$ describing the probability that the system will evolve to $\mathbf{z}'_{A}$ given that it started in state $\mathbf{z}_A$.  Clearly, we require $\eta(\mathbf{z}'_{A}|\mathbf{z}_A)\ge 0$ for all $\mathbf{z}_{A},\mathbf{z}_{A'}$, and $\int \textrm{d}\mathbf{z}'_{A} \eta(\mathbf{z}'_{A}|\mathbf{z}_A) =1$ for all $\mathbf{z}_{A}$.  If an agent's knowledge of the dynamics is described by $\eta(\mathbf{z}'_{A}|\mathbf{z}_A)$, and their knowledge of the initial state is described by the epistemic state $\mu(\mathbf{z}_A)$, then their knowledge of the final state will be described by the epistemic state
\begin{equation}
\mu'(\mathbf{z}'_{A})=\int \textrm{d}\mathbf{z}_{A} \eta(\mathbf{z}'_{A}|\mathbf{z}_A) \mu(\mathbf{z}_A).
\end{equation}
%Not every set of transition probabilities is allowed in the theory.
%For instance, consider the choice $\eta(\mathbf{z}_{A'}|\mathbf{z}_A) =\delta_{q'_{A},q_{A}}\delta_{p'_{A},-p_{A}}$.  This corresponds to the deterministic map $q_{A'}=q_{A}$ and $p_{A'}=-p_{A}$, that is, momentum reversal (or equivalently, time reversal).  This cannot arise as a symplectic transformation because the Poisson bracket is not preserved.  Furthermore, although one could conceive of implementing this map by measuring both the position and the momentum and then re-preparing the system with an inverted momentum, such a measurement is forbidden in the theory.  So this particular set of transition probabilities cannot describe an agent's knowledge of the dynamics.

We can also represent the transformation of the agent's knowledge by a \emph{transfer functional,} that is, a linear map over functions on phase space $\Gamma_{A} :L(\mathcal{M}_A)\rightarrow L(\mathcal{M}_{A})$, specifically,
\begin{equation}
  \Gamma_{A} [f](\mathbf{z}'_{A})=\int d\mathbf{z}_A\, \eta(\mathbf{z}'_{A}|\mathbf{z}_A)f(\mathbf{z}_A)\,.
\end{equation}
This map is \emph{norm-preserving}, that is, it satisfies $|\Gamma_{A} [f]|=|f|$, for all functions $f \in L(\mathcal{M}_A)$.  It is also \emph{positivity-preserving}, which is to say that if $f \in L_+(\mathcal{M}_A)$ then $\Gamma_{A} [f] \in L_+(\mathcal{M}_{A'})$).

The question is: which transition probabilities, or equivalently, which transfer functionals are valid within ERL mechanics?

%\rws{We cannot assume that all transfer functionals that are ``completely-validity-preserving'' are allowed!  We need to show that these are the reductions of linear symplectic transformations on a larger system, or mixtures of linear symplectic transformations.}

A necessary condition on the set of valid transformations on a system is that, viewed as transfer functionals, they must take valid epistemic states on the system to valid epistemic states, that is, they must be \emph{validity-preserving}.  But it is also necessary that when acting on part of a larger system, they also take valid epistemic states on that larger system to valid epistemic states; we say that they are \emph{completely validity-preserving} or CVP (in analogy with the property of maps in quantum theory of being completely positivity-preserving).
%, expressed as norm-preserving transfer functions, are those that are ``complete validity preserving'' (CVP), analogous to completely positive maps in quantum theory, in the sense that they not only take the set $L_{\rm valid} (\mathcal{M}_A)$ of valid epistemic states on $\mathcal{M}_A$ to valid epistemic states on $\mathcal{M}_{A'}$, but if acted on part of a larger system $\mathcal{M}_A \times \mathcal{M}_B$, they take the set of valid epistemic states $L_{\rm valid} (\mathcal{M}_A \times \mathcal{M}_B)$ on that larger system to valid epistemic states on $\mathcal{M}_{A'} \times \mathcal{M}_B$.
Thus, we require that if
\begin{equation}
  \mu_{AB}(\mathbf{z}_A,\mathbf{z}_B)\in L_{\rm valid} (\mathcal{M}_A \times \mathcal{M}_B)\,,
\end{equation}
then
\begin{equation}
  \int d\mathbf{z}_A \,\eta(\mathbf{z}'_{A}|\mathbf{z}_A)\mu_{AB}(\mathbf{z}_A,\mathbf{z}_B)\in L_{\rm valid} (\mathcal{M}_{A} \times \mathcal{M}_B)\,.
\end{equation}
Defining the identity transfer functional $\mathrm{id}: \mathcal{M} \to \mathcal{M}$ by $\mathrm{id}[f]=f$, we can express the condition of a transfer functional $\Gamma_{A}$ being CVP compactly as
\begin{equation}
  (\Gamma_{A} \otimes \textrm{id}_B)[\mu_{AB}] \in L_{\rm valid}(\mathcal{M}_{A'} \times \mathcal{M}_B)\,.
\end{equation}
%where $\Gamma_{A'|A} \times I_B$ denotes the transfer functional given by the Cartesian product of $\Gamma_{A'|A}$ acting on $\mathcal{M}_A$ and the identity on $\mathcal{M}_B$.

The other condition that a transformation must satisfy in order to be considered valid is that it must supervene on valid ontic dynamics -- either the transformation corresponds to linear symplectic evolution on the system's phase space (the reversible case) or it must correspond to adjoining to the system an ancillary system prepared according to a valid epistemic state, coupling the pair via a linear symplectic evolution on the joint phase space, and then marginalizing over the ancillary system.  If this condition holds, we say that the transformation satisfies \emph{ontic supervenience}.

To see why this condition is important, it suffices to note that a transformation may be validity-preserving but not satisfy the ontic supervenience property. Momentum reversal (or equivalently, time reversal) is such a transformation.  It is defined by the conditional $\eta(\mathbf{z}'_{A}|\mathbf{z}_A) =\delta_{q'_{A},q_{A}}\delta_{p'_{A},-p_{A}}$, which corresponds to the deterministic map $q'_{A}=q_{A}$ and $p'_{A}=-p_{A}$.  This cannot arise as a symplectic transformation on the system because the Poisson bracket is not preserved.  Furthermore, although one could conceive of implementing this map by measuring both the position and the momentum and then re-preparing the system with an inverted momentum, such a measurement is forbidden in the theory.
%So this particular set of transition probabilities cannot describe an agent's knowledge of the dynamics.
So, while momentum reversal takes every valid epistemic state to a valid epistemic state, it does not satisfy ontic supervenience.\footnote{In the context of the Spekkens toy theory, the ``universal state inverter'' transformation, which takes every epistemic state of a single elementary system to the epistemic state that has disjoint support with it, is an example of a transformation that is validity-preserving but does not supervene on the ontic dynamics, as discussed in Sec.~III.C of Ref.~\cite{Spe07}. If one tries to  supplement the toy theory with such transformations, as is considered in Ref.~\cite{Sko08}, one is left with a theory that no longer admits of a straightforward realist interpretation.}

Nonetheless, we will show that every transformation that is completely validity-preserving satisfies the ontic supervenience property and so is a valid transformation.

We will also show that one can characterize the set of valid transformations by their action on the perfectly correlated state.
\begin{proposition}[Valid transformations]\label{lem:validtransfer}\label{thm:validtransferfunctions}
A transformation on a system is valid if and only if the bipartite epistemic state one obtains by acting
it on half of a perfectly correlated state of a pair of such systems is valid.  In other words, if $\Gamma_A$ denotes the transfer functional on $L(\mathcal{M}_A)$, $\mu^{\mathrm{corr}}_{AB}$ denotes the perfectly correlated state on a pair of identical systems $\mathcal{M}_A \times \mathcal{M}_B$, and
\begin{equation}
  \label{eq:CJisomorphism}
  \mu_{AB}^{\Gamma }\equiv (\Gamma_{A}\otimes \textrm{id}_B)[\mu_{AB}^{\mathrm{corr}}]\,.
\end{equation}
then $\Gamma_A$ is valid if and only if $\mu_{AB}^{\Gamma }\in L_{\rm valid}(\mathcal{M}_A \times \mathcal{M}_B)$.
\end{proposition}
Note that $\mu_{AB}^{\Gamma }$ has the same marginal on $B$ as $\mu_{AB}^{\mathrm{corr}}$, that is, a uniform marginal.  Therefore, the valid transfer functionals on $A$ are in one-to-one correspondence with the valid epistemic states on a pair of copies of $A$ that have a uniform marginal on one of the copies.  We shall say simply that the valid transformations are isomorphic to valid bipartite states. This isomorphism is the analogy within ERL mechanics of the Choi isomorphism in quantum theory.

The rest of the section will seek to prove these results.  The strategy of the proof is to demonstrate (i) that a transfer functional is completely validity-preserving if and only if it is isomorphic to a valid bipartite state, and (ii) that a transfer functional satisfies the ontic supervenience property if and only if it is isomorphic to a valid bipartite state.  Together, these two facts imply that \emph{every} transformation that is completely validity-preserving satisfies the ontic supervenience property.  Therefore, the condition of being completely validity-preserving is not only necessary for a transformation to be valid but sufficient as well (unlike the condition of being validity-preserving, which is not sufficient). Given this characterization of the valid transformations, proposition \ref{lem:validtransfer} then follows from (i).

%%%%%%%%%%%%%

%Let $\Gamma_{A'|A}$ be a valid norm-preserving transfer function.  Using this map, along with the correlated state $\mu^{\mathrm{corr}}_{AB}$ on a pair of identical systems $\mathcal{M}_A \times \mathcal{M}_B$, we define an associated bipartite epistemic state $\mu^\Gamma \in L_{\rm valid}(\mathcal{M}_A \times \mathcal{M}_B)$ as
%\begin{equation}
%  \label{eq:CJisomorphism}
%  \mu_{A'B}^{\Gamma }=(\Gamma_{A'|A}\times I_B)[\mu_{AB}^{\mathrm{corr}}]\,.
%\end{equation}
%As we will prove below, this map is an isomorphism (recalling that $A$ and $B$ are identical systems).  That is, valid transfer functions are in one-to-one correspondence with valid epistemic bipartite states.  This isomorphism is the analogy to the Choi-Jamiolkowski isomorphism used in quantum theory.

%\rws{Given the map we can generate the corresponding state, and given the state, we can physically implement the map with some probability (this is the analogue of teleporting $\rho $ through the gate $\mathcal{E}$ - see Gottesman and Chuang, nature 402, 390 (1999)$)$. \ Show how this is done. } \sdb{Need to check what Matt has worked out.}

We begin by establishing the connection between the CVP property and the isomorphism property.  To do so, it is useful to note a general analogue of quantum teleportation (and entanglement swapping) within ERL mechanics (formalizing the discussion in Sec.~\ref{sec:Teleportation}).  We begin by defining a functional that represents marginalizing over $\mathcal{M}_A$, namely, $\mathrm{Tr}_A: \mathcal{M}_A \to \mathbb{R}$ defined by $\mathrm{Tr}_A[f] = \int \textrm{d}\mathbf{z}_A f(\mathbf{z}_A) = |f|$ (the notation is chosen to be suggestive of the analogous quantum trace operation).
\begin{lemma}[Teleportation]\label{lemma:swap}
Any epistemic state $\mu_{AB}$ on $\mathcal{M}_{A}\times \mathcal{M}_{B}$
satisfies the relation
\begin{equation}
\mu _{AB} \propto \textrm{Tr}_{CD} [\xi_{CD}^{\mathrm{corr} } \mu _{AC}^{\mathrm{corr}} \mu _{DB}]\,,
%\mu _{AB} \propto \Bigl[(\textrm{id}_{A}\otimes \textrm{Tr}_{CD} \otimes
%\textrm{id}_{B})[\xi_{CD}^{\mathrm{corr} } \mu _{AC}^{\mathrm{corr}} \mu _{DB}]\Bigr]\,,
\end{equation}
where $C,D$ are ancillary systems identical to $A$, and $\xi_{CD}^{\mathrm{\
corr}} \propto \mu _{CD}^{\mathrm{corr}}$ is the indicator function associated with the maximally-correlated state.
\end{lemma}

\begin{proof}
We make use of the explicit expression for $\mu^{\mathrm{corr}}$ given in
Eqs.~(\ref{eq:corr}-\ref{eq:corr3}),
%and the equality $\mu^{\mathrm{corr}} = \xi^{\mathrm{corr}}/|\xi^{\text{corr}}|$, which yields
and the proportionality $\xi^{\mathrm{corr}} \propto \mu^{\mathrm{corr}}$, to obtain
\begin{align}
\textrm{Tr}_{CD} [&\xi_{CD}^{\mathrm{corr} } \mu _{AC}^{\mathrm{corr}} \mu _{DB}] \notag \\
%\Bigl[(I_{A}&\times \xi_{CD}^{\mathrm{corr} }\times I_{B})[\mu _{AC}^{\mathrm{corr}}\times \mu %_{DB}]\Bigr](\mathbf{z}_A,\mathbf{z}_B)  \notag \\
&\propto\prod_i \int dq_{iC}\,dp_{iC}\,dq_{iD}\,dp_{iD}\, \notag \\
&\qquad \times \delta(q_{iC}-q_{iD}) \delta(p_{iC}+p_{iD})\delta(q_{iA}-q_{iC}) \notag \\
&\qquad \times \delta(p_{iA}+p_{iD})\mu_{DB}(\mathbf{z}_D,\mathbf{z}_B)  \notag \\
&\propto \mu_{AB}(\mathbf{z}_A,\mathbf{z}_B) \,,
\end{align}
\end{proof}

%We can now prove that the map~(\ref{eq:CJisomorphism}) is an isomorphism:
We can now prove the first lemma concerning valid transformations.
\begin{lemma}[CVP and isomorphism]\label{lem:CVPIso}
A transformation on a system is completely validity-preserving if and only if the bipartite epistemic state one obtains by acting it on half of a perfectly correlated pair of such systems is valid.
\end{lemma}

\begin{proof} Necessity is trivial to prove.  If a transfer function $\Gamma_{A}$ is CVP, then it maps all valid epistemic states on $\mathcal{M}_A \times \mathcal{M}_B$ to valid epistemic states on $\mathcal{M}_{A} \times \mathcal{M}_B$.
%, that is, for all $\mu_{AB}\in L_{\rm valid}(\mathcal{M}_A \times \mathcal{M}_B)$,
%\begin{equation}
%  (\Gamma _{A'|A}\times I_{B})[\mu_{AB}] \in L_{\rm valid}(\mathcal{M}_{A'} \times \mathcal{M}_B)\,.
%\end{equation}
Because $\mu_{AB}^{\mathrm{corr}}$ is a valid epistemic state on $\mathcal{M}_A \times \mathcal{M}_B$, then $\mu_{AB}^{\Gamma} = (\Gamma_{A}\otimes \textrm{id}_{B})[\mu_{AB}^{\mathrm{corr}}]$ is as well.

To prove sufficiency, we must show that any $\Gamma_{A}$ satisfying
\begin{equation}
(\Gamma _{A}\otimes \textrm{id}_{B})[\mu _{AB}^{\mathrm{corr}}]\in L_{\rm valid}(\mathcal{M}_{A} \times \mathcal{M}_B)\,,
\label{eq:TransFunctConstraint}
\end{equation}
also satisfies $(\Gamma _{A}\otimes \textrm{id}_{B})[\mu_{AB}]\in L_{\rm valid}(\mathcal{M}_{A} \times \mathcal{M}_B)$ for all $\mu_{AB}\in L_{\rm valid}(\mathcal{M}_A \times \mathcal{M}_B)$.

Using Lemma~\ref{lemma:swap}, we now calculate the action of $(\Gamma _{A}\otimes \textrm{id}_{B})$
on an arbitrary state $\mu_{AB} \in L_{\rm valid}(\mathcal{M}_A \times \mathcal{M}_B)$.
\begin{align}
  (\Gamma_{A} & \otimes \textrm{id}_{B})[\mu_{AB}]  \notag \\
  &= \Bigl[(\Gamma_{A}\otimes \textrm{Tr}_{CD} \otimes \mathrm{id}_{B}) [\xi_{CD}^{\mathrm{corr}} \mu_{AC}^{\mathrm{corr}} \mu _{DB}]\Bigr]  \notag \\
  & =\Bigl[(\mathrm{id}_{A} \otimes \textrm{Tr}_{CD} \otimes \mathrm{id}_{B})[\xi_{CD}^{\mathrm{corr}} \mu_{AC}^{\Gamma} \mu_{DB}]\Bigr]\,.
\end{align}
Because $\mu_{AC}^{\Gamma} \mu _{DB}$ is a valid epistemic state (being a product of two valid epistemic states) and because $\xi_{CD}^{\mathrm{corr}}$ is a valid indicator function, the result is a valid epistemic state.
\end{proof}

Next, we need to establish that a transformation satisfies the ontic supervenience property if and only if it is isomorphic to a valid bipartite state.  We begin by characterizing what the ontic supervenience property implies about how the transformation acts on the covariance matrix.

\begin{lemma}\label{lemma:irreversible}
A transformation on $A$ satisfies the ontic supervenience property (i.e. it can be realized by coupling to an environment via a joint linear symplectic transformation followed by marginalization) if and only if the covariance matrix on $A$ transforms as
% $\gamma_A \to \gamma'_A$, where
\begin{equation} \label{eq:howCMevolves}
  \gamma_A \mapsto X^T \gamma_A X + Y\,,
\end{equation}
for real matrices $X,Y$ that satisfy
\begin{equation}
  \label{eq:XYcondition}
  Y \geq i\Sigma_A - i X^T \Sigma_A X \,.
\end{equation}
\end{lemma}

%Finally, we prove that all valid transfer functions $\Gamma_{A'|A}$ arise as the result of a transformation of the ontic states.  This ontic state space dynamics need not be reversible; we consider a more general dynamics that can arise by coupling the system to another (an `environment'), performing some reversible evolution on the joint phase space (a symplectic transformation), and finally discarding the environment.
This result follows in a straightforward manner from previous work on unitary dilations of Gaussian quantum channels \cite{Caruso}. Nonetheless, for clarity, we repeat some of the details of the proof here.

\begin{proof} Consider necessity first. We begin by describing how transformations that satisfy the ontic supervenience property act on the covariance matrix.  Consider a system $A$, and a valid epistemic state $\mu_{A}$ with covariance matrix $\gamma_{A}$, initially uncorrelated with an environment $E$, described by a valid epistemic state $\mu_E$ and covariance matrix $\gamma_{E}$.  Because they are initially uncorrelated, the covariance matrix of the joint system $AE$ is $\gamma_{AE} = \text{diag}(\gamma_A,\gamma_E)$.  The pair of systems is then acted upon by a linear symplectic transformation $S$ on the joint phase space $AE$, describing a general reversible transformation.  This matrix can be expressed in block form as
\begin{equation}
  S = \begin{pmatrix} S_{AA} & S_{AE} \\ S_{EA} & S_{EE} \end{pmatrix}\,,
\end{equation}
and satisfies $S^T \Sigma_{AE} S = \Sigma_{AE}$.  The covariance matrix $\gamma'_{AE}$ for the joint system after this transformation is given by $\gamma'_{AE} = S^T \gamma_{AE} S$.  Considering only the resulting marginal distribution $\mu'_A$ on the system after the transformation (marginalizing over the environment), the resulting covariance matrix $\gamma'_A$ of $\mu'_A$ is
\begin{equation} \label{eq:gammaprime}
  \gamma'_A = S^T_{AA}\gamma_A S_{AA} + S^T_{EA}\gamma_E S_{EA} \,.
\end{equation}
Note that, because $\mu_A$ and $\mu_E$ are both valid epistemic states, and we applied valid operations (a reversible linear symplectic transformation, and a marginalization), then the final marginal $\mu'_A$ will be valid, i.e., $\gamma'_A+i\Sigma_A \ge 0$.
Eq.~\eqref{eq:gammaprime} shows that the covariance matrix transforms as Eq.~\eqref{eq:howCMevolves} prescribes, with $X = S_{AA}$ and $Y = S^T_{EA}\gamma_E S_{EA}$.  These are both real matrices, but it remains to show that they satisfy Eq.~\eqref{eq:XYcondition}.
%Note that, because $\mu_A$ and $\mu_E$ are both valid epistemic states, and we applied valid operations (a reversible linear symplectic transformation, and a marginalization), then the final marginal $\mu'_A$ will be valid, i.e., $\gamma'_A+i\Sigma_A \ge 0$.  Using the expression for $\gamma'_A$ provided in Eq.~\eqref{eq:gammaprime}, we obtain $S^T_{AA}\gamma_A S_{AA} + S^T_{EA}\gamma_E S_{EA} \ge 0$.  It follows that the transformation   satisfies the CUP.
%We seek to prove that any linear symplectic  coupling to a Gaussian state followed by marginalization leads to a transformation of the Gaussian state of the form $\gamma\mapsto X^{T}\gamma X+Y$ where $Y\geq i\Sigma_{A}-iX^{T}\sigma_{2n}X.$

Because $\gamma_{E}+i\Sigma_{E}\geq0$ it follows that
\begin{equation}
S_{EA}^{T}\left(  \gamma_{E}+i\Sigma_{E}\right)  S_{EA}\geq0.\label{eq:SEA}
\end{equation}
Using the fact that $S\Sigma S^{T}=\Sigma,$ where $\Sigma=\mathrm{diag}\left(
\Sigma_{A},\Sigma_{E}\right)  ,$ we infer that $S_{AA}^{T}\Sigma_{A}%
S_{AA}+S_{EA}^{T}\Sigma_{E}S_{EA}=\Sigma.$ \ Substituting this into Eq.~\eqref{eq:SEA}, we
have
\[
S_{EA}^{T}\gamma_{E}S_{EA}+i\left(  \Sigma-S_{AA}^{T}\Sigma_{A}S_{AA}\right)
\geq0.
\]
In other words,%
\[
Y\geq-i\Sigma_{A}+iX^{T}\Sigma_{A}X,
\]
from which Eq.~\eqref{eq:XYcondition} can be obtained by taking the complex
conjugate.

To prove sufficiency, one must show that it is possible to  find a
symplectic matrix $S$ and a covariance matrix $\gamma_{E}$ leading to any $X$
and $Y$ that satisfy $Y \geq i\Sigma_{A} - iX^{T}\Sigma_{A}X$. The construction is somewhat involved, so we do not repeat it here, but simply refer the reader to \cite{Caruso}.
\end{proof}

Having characterized the transformations that satisfy the ontic supervenience property
by how they act on the covariance matrix, we now demonstrate that all such transformations are isomorphic to a valid bipartite state.

%We now prove that any valid transfer function $\Gamma_{A'|A}$ can be expressed as a map of this form.  To do so, we make use of Lemma~\ref{lem:validtransfer}, which provides an isomorphism between valid transfer functions and valid bipartite states.

\begin{lemma}
%A transfer functional $\Gamma_A$
A transformation acts on covariance matrices in the manner described in lemma~\ref{lemma:irreversible} if and only if the bipartite epistemic state one obtains by acting it on half of the perfectly correlated state for a pair of such systems is valid.
%Any valid bipartite state $\mu_{AB}$ having a uniform marginal on $B$ can be realised by acting on system $A$ of the perfectly correlated state $\mu^{\rm corr}_{AB}$ with an irreversible transformation as defined in Lemma~\ref{lemma:irreversible}.
\end{lemma}

\begin{proof}
Necessity is trivial to prove.  By lemma~\ref{lemma:irreversible}, the transformation of interest corresponds to coupling to an ancilla in a valid state by a linear symplectic transformation and marginalizing over the ancilla.  Given that the perfectly correlated state is a valid state, and given that every part of this transformation clearly keeps one within the set of valid states, the final bipartite epistemic state will be valid.

It remains to prove sufficiency. We assume that the transfer functional $\Gamma_A$ satisfies the isomorphism property, that is, that the bipartite epistemic state resulting from the transformation acting on half of the perfectly correlated state, denoted $\mu^{\Gamma}_{AB}$ and having a uniform marginal on $B$, is valid. The latter is described by its means $\mathbf{d}_A$ on $A$ (the means on $B$ are not well-defined because the distribution is uniform on $B$) and by its covariance matrix
\begin{equation}
  \label{eq:ArbitraryAB}
  \gamma_{AB} = \begin{pmatrix} \gamma_{A} & C \\ C^T & \gamma^{\rm uniform}_B \end{pmatrix} \,.
\end{equation}
 The assumption that $\mu^{\Gamma}_{AB}$ is a valid epistemic state places no restriction on $\mathbf{d}_A$, but it does restrict $\gamma_{AB}$ to satisfy the CUP condition, $\gamma_{AB} + i\lambdabar \Sigma_{AB} \geq 0$.
Making use of the result from linear algebra described in Eq.~\eqref{eq:partitionedmatrix}, we infer that the CUP condition on $\gamma_{AB}$ is equivalent to the condition
\begin{equation}
  \label{eq:ArbitraryABcondition}
  \gamma_{A} + i\lambdabar \Sigma_{A} \geq C[\gamma^{\rm uniform}_B+i\lambdabar \Sigma_B]^{-1} C^T \,.
\end{equation}
This alternative form will be useful in what follows.
%Given that the uniform distribution is valid, meaning $\gamma^{\rm uniform}_B +i\lambdabar \Sigma_B \geq 0$, the positivity condition for a block matrix of this form can be expressed as
%\begin{equation}
%  \label{eq:ArbitraryABcondition}
%  \gamma_{A} + i\lambdabar \Sigma_{A} \geq C[\gamma^{\rm uniform}_B+i\lambdabar \Sigma_B]^{-1} C^T \,.
%\end{equation}

We need to show that $\mu^{\Gamma}_{AB}$ being a valid epistemic state implies that $\Gamma_A$ acts in the manner described in Lemma~\ref{lemma:irreversible}.  It suffices to show that \emph{every} valid epistemic state on $AB$ can be obtained from the perfectly correlated epistemic state by \emph{some} transfer functional of this form.
%\textbf{Can we express the logic as follows?  The isomorphism is such that every bipartite state with uniform marginal on B defines a transfer functional (via gate teleportation).  We want to know whether every transfer functional defined by a valid bipartite state acts on the covariance matrix in accordance with the lemma.}

%We need to demonstrate that a transformation of the form described by Lemma~\ref{lemma:irreversible} can, when acting on the perfectly correlated state, yield an arbitrary state of this form.  Given that an arbitrary vector of means $\mathbf{d}_A$ can be obtained with a symplectic transformation only on $A$, the key criterion is whether we can obtain a covariance matrix of the form of Eq.~(\ref{eq:ArbitraryAB}), i.e. an arbitrary pair of matrices $\gamma_{A}$ and $C$ satisfying Eq.~(\ref{eq:ArbitraryABcondition}).

First, we consider the covariance matrix of the perfectly correlated state $\mu_{AB}^{\mathrm{corr}}$.  This state is only defined as the limit of a squeezed Gaussian state, as in Eq.~(\ref{eq:corr}), and so in the following argument we consider finite squeezing throughout, and only take the limit in the final stages of our argument.  It is convenient to change our squeezing parameter $s$, for which $s\to 0$ is the desired limit, to be reparametrised as $s=\exp(-r)$, and thus $r \to \infty$ is the desired limit.  With this substitution, it is straightforward to show that the covariance matrix $\gamma^{\mathrm{corr}}_{AB}$ of the perfectly correlated state $\mu_{AB}^{\mathrm{corr}}$ is the $r\to \infty$ limit of
\begin{equation}
  \gamma^{\mathrm{corr}}_{AB}(r) = \begin{pmatrix} D_+(r) & D_-(r) \\ D_-(r) & D_+(r) \end{pmatrix} \,,
\end{equation}
where $D_{\pm}(r)$ are diagonal real matrices defined by
\begin{align}
  D_+(r) &= \cosh(2r)\, \text{diag}(1,\lambdabar^2,1,\lambdabar^2,\ldots)\,, \\
  D_-(r) &= \sinh(2r)\, \text{diag}(1,-\lambdabar^2,1,-\lambdabar^2,\ldots)\,.
\end{align}
We note that the marginals on $A$ and $B$ have covariance matrix $D_+(r)$ which, as $r \to \infty$, is the uniform distribution (as a limit of a Gaussian).

By acting on system $A$ of the pair of systems $AB$, initially described by the perfectly correlated state $\mu_{AB}^{\mathrm{corr}}$, with a general transformation of the form described in Lemma~\ref{lemma:irreversible}, the resulting state $\mu_{AB}$ has covariance matrix
\begin{equation}
  \label{eq:CJstateXY}
  \gamma_{AB}(r) = \begin{pmatrix} X^T D_+(r) X + Y & X^T D_-(r) \\ D_-(r)X & D_+(r) \end{pmatrix} \,.
\end{equation}
where the matrices $X$ and $Y$ must satisfy Eq.~(\ref{eq:XYcondition}).  We then wish to show that $X$ and $Y$ can be chosen such as to produce any state of the form of Eq.~(\ref{eq:ArbitraryAB}) satisfying Eq.~(\ref{eq:ArbitraryABcondition}).  We want $X$ and $Y$ such that
\begin{align}
  \label{eq:CfunctXY}
  C &= X^T D_-(r) \,, \\
  \label{eq:gammafunctXY}
  \gamma_{A} &= X^T D_+(r) X + Y \,.
\end{align}
As $D_-(r)$ is invertible, we can choose
\begin{align}
  X &= D_-(r)^{-1} C^T \,, \\
  Y &= \gamma_{A} - C D_-(r)^{-1} D_+(r) D_-(r)^{-1} C^T \,.
\end{align}
It remains to be shown whether $X$ and $Y$ can be chosen as such, because they are constrained to satisfy the condition given by Eq.~(\ref{eq:XYcondition}).  We now show that they can, by demonstrating that the condition given by Eq.~(\ref{eq:XYcondition}) is equivalent to the condition given by Eq.~(\ref{eq:ArbitraryABcondition}).

Recall that $\gamma^{\rm uniform}_B = \lim_{r\to\infty} D_+(r)$.  We will substitute $D_+(r)$ for $\gamma^{\rm uniform}_B$ in Eq.~(\ref{eq:ArbitraryABcondition}), and take the $r\to\infty$ limit at the final step.
With the substitutions given by Eqs.~(\ref{eq:CfunctXY}-\ref{eq:gammafunctXY}), the condition of Eq.~(\ref{eq:ArbitraryABcondition}) is expressed as
\begin{multline}
  X^T D_+(r) X + Y + i \lambdabar \Sigma_A \\
  \geq X^T D_-(r) [D_+(r) +i\lambdabar \Sigma_B]^{-1} D_-(r) X \,.
\end{multline}
Rearranging gives
\begin{multline}
  Y \geq -i\lambdabar \Sigma_A \\
  - X^T \bigl( D_+(r) - D_-(r)\bigl[ D_+(r) + i\lambdabar \Sigma_B \bigr]^{-1} D_-(r) \bigr) X \,.
\end{multline}
We then make use of the following identity, which holds for all $r$:
\begin{multline}
  D_+(r) - D_-(r)\bigl[ D_+(r) + i\lambdabar \Sigma_B \bigr]^{-1} D_-(r) \\
  = - i \lambdabar \Sigma_A \,.
\end{multline}
We thereby obtain
\begin{equation}
  Y \geq -i\lambdabar \Sigma_A + i \lambdabar X^T \Sigma_A X \,.
\end{equation}
By taking the complex conjugate of this equation, we recover Eq.~(\ref{eq:XYcondition}).  Therefore, in the limit $r \to \infty$, where $D_+(r)$ becomes the uniform distribution $\gamma^{\rm uniform}_B$, we have proved the equivalence of the conditions of Eqs.~(\ref{eq:ArbitraryABcondition}) and (\ref{eq:XYcondition}).
\end{proof}

As a final comment on transformations, note that we do not need to separately specify how the epistemic state of a system updates as the result of a measurement in ERL mechanics; this follows from cases we have already considered.  By the assumption that ERL mechanics is just classical mechanics with an epistemic restriction, every measurement must be understandable as adjoining some degrees of freedom of an apparatus to the system, coupling these by a linear symplectic transformation, followed by acquiring information about the apparatus. For any valid set of indicator functions on the system, one can achieve a measurement associated with this set by such a procedure.  The argument follows a similar logic to the ERL-mechanical analogue of von Neumann's dynamical analysis of measurement, presented in Sec.~\ref{sec:motilityofthecut}.  The update needn't always be analogous to the projection postulate however.  The manner in which the epistemic state of the system updates depends on the particular manner in which the measurement is implemented. (This is analogous to how, in quantum theory, there are many state update rules associated with a given POVM; even for projective measurements the projection postulate is just one of the possibilities.)  We do not need to consider this situation afresh because both components of the process have been considered already: how the bipartite epistemic state of a pair of systems transforms under a linear symplectic transformation, and how the epistemic state of a system updates as a result of a measurement on another system with which it is correlated.

\subsection{Operational formulation of Gaussian quantum mechanics}
\label{sec:GaussianQM}

We review the Wigner representation, and then proceed to define Gaussian quantum mechanics.  For further reading on these topics, see Ref.~\cite{GardinerZoller,Wee11}.

\subsubsection{The Wigner representation}
\label{sec:Wigner}

In the Wigner representation~\cite{Wig32,GardinerZoller}, quantum states are represented as real-valued functions over phase-space that integrate to unity.  Specifically, for a system of $n$ canonical degrees of freedom and following the notation of Sec.~\ref{sec:QuantumIntro}, the Wigner representation for a quantum state $\rho$ is
\begin{equation}
  W_{\rho}(\mathbf{z})=\mathrm{Tr}(\rho A_{\mathbf{z}})\,,
\end{equation}
where
\begin{equation}
  A_{\mathbf{z}}=\bigotimes_{i=1}^{n}A_{z_{i}}\,,
\end{equation}
and
\begin{equation}
  A_{z_{i}}= \frac{1}{\pi \hbar} \int {\rm d}y\, e^{-ip_{i}y/\hbar} | q_{i}-
  \tfrac{1}{2}y \rangle \langle q_{i}+\tfrac{1}{2}y| \,,
\end{equation}
and $|q\rangle$ is the position eigenstate.  We note that these operators satisfy $\mathrm{Tr}(A_{\mathbf{z}})=\frac{1}{\pi \hbar}$.

The operators $A_{z_i}$ satisfy the identity
\begin{equation}
  \label{eq:HSinnerproductinWigrep}
  \mathrm{Tr}\left( AB\right) = (\pi \hbar)^n \int \mathrm{d}\mathbf{z}\,\mathrm{Tr}\left( AA_{
  \mathbf{z}}\right) \mathrm{Tr}(BA_{\mathbf{z}})\,.
\end{equation}
This identity follows from the fact that the $A_{\mathbf{z}}$,
considered as vectors in the Hilbert-Schmidt operator space, form a
resolution of unity.  Denoting $\mathrm{Tr}(AB)$ as an inner
product on the Hilbert-Schmidt operator space, $\left\langle A|B\right\rangle$, we have simply used the
fact that $\left\langle A\right\vert (\int \mathrm{d}\mathbf{z}\left\vert A_{
\mathbf{z}}\right\rangle \left\langle A_{\mathbf{z}}\right\vert)
\left\vert B\right\rangle =\left\langle A|B\right\rangle$.

The Wigner representation of a Hermitian operator $O$ is the real-valued function
$W_{O}(\mathbf{z})=\mathrm{Tr}(O A_{\mathbf{z}}),$ and the expectation value of $O$ in state $\rho$ is recovered by the Euclidean inner product of the Wigner representations of $O$ and $\rho$,
\begin{align}
  \label{eq:WignerObservables}
  (\pi \hbar)^n \int \mathrm{d}\mathbf{z}\,W_{\rho}(\mathbf{z})W_{O}(\mathbf{z})
  &=(\pi \hbar)^n \int \mathrm{d}\mathbf{z}\,\mathrm{Tr}(\rho A_{\mathbf{z}})
  \mathrm{Tr}(OA_{\mathbf{z}}) \nonumber \\
  &=\mathrm{Tr}(\rho O)\,,
\end{align}
where we have used the identity~(\ref{eq:HSinnerproductinWigrep}).

The most general measurement allowed by quantum theory, associated
with a POVM $\{E_{\mathbf{y}}\}$, also admits a Wigner representation as a
set of real-valued functions over phase space that sum to the uniform measure over the phase space.  Specifically, we have
\begin{equation}
  W_{E_{\mathbf{y}}}(\mathbf{z})=(\pi \hbar)^n \mathrm{Tr}(E_{\mathbf{y}} A_{\mathbf{z}})\,,
\end{equation}
which, given that $\int \mathrm{d}{\mathbf{y}} E_{\mathbf{y}}=I$, implies that
\begin{equation}
  \int \mathrm{d}{\mathbf{y}} W_{E_{\mathbf{y}}}(\mathbf{z})=(\pi \hbar)^n\mathrm{Tr}(A_{\mathbf{z}})=1\,.
\end{equation}
The general form of the Born rule, which asserts that given a
preparation associated with quantum state $\rho$, and a
measurement associated
with POVM $\{E_{\mathbf{y}}\}$ the probability density for outcome ${\mathbf{y}}$, $\textrm{Tr}(\rho E_{\mathbf{y}})$, is recovered in the Wigner representation as the Euclidean inner product of the Wigner representation of the quantum
state with that of the POVM element associated with $\mathbf{y}$,
\begin{equation}
  \int \mathrm{d}\mathbf{z}\,W_{\rho}(\mathbf{z})W_{E_{\mathbf{y}}}(\mathbf{z})=
  \mathrm{Tr}(\rho E_{\mathbf{y}})\,,  \label{eq:Bornrule}
\end{equation}
where we have again used Eq.~(\ref{eq:HSinnerproductinWigrep}).

For certain unitary operations, such as displacement and squeezing
operations, it is well-known how to determine their effect
within the Wigner representation.
%We will return to these, but it is useful to first
However, we must consider how the most general transformation, associated
with a completely-positive trace-nonincreasing linear map
$\mathcal{E}$, is represented in the Wigner representation.  (To our knowledge, this result has
not previously been made explicit in the literature on the Wigner representation).  Such a
transformation can be represented by a real-valued function over two copies of the phase space, denoted
 $W_{\mathcal{E}}(\mathbf{z}|\mathbf{z}')$, which satisfies
\begin{equation}
  W_{\mathcal{E}(\rho)}(\mathbf{z})=\int
  \mathrm{d}\mathbf{z}'\, W_{\rho}(\mathbf{z}')W_{\mathcal{E}}(\mathbf{z}|\mathbf{z}')\,,
  \label{eq:Wigneractionofmap}
\end{equation}
and has the form
\begin{equation}
  W_{\mathcal{E}}(\mathbf{z}|\mathbf{z}')=(\pi \hbar)^n\mathrm{Tr}\left( A_{\mathbf{z}}\mathcal{E}(A_{\mathbf{z}'})\right)\,.
\end{equation}
This result can be proved as follows.  Using the identity of Eq.~(\ref{eq:HSinnerproductinWigrep}) and the definition of the Hermitian adjoint $\mathcal{E}^{\dag}$ of $\mathcal{E}$, namely, $\mathrm{Tr}(\mathcal{E}(A)B)=\mathrm{Tr}(A\mathcal{E}^{\dag}(B))$, we can infer that
\begin{align}
  W_{\mathcal{E}(\rho)}(\mathbf{z})
  & =\mathrm{Tr}\left( \mathcal{E}(\rho)A_{\mathbf{z}}\right)  \\
  & =\mathrm{Tr}\left( \rho\mathcal{E}^{\dag }(A_{\mathbf{z}})\right)  \\
  & =(\pi \hbar)^n\int \mathrm{d}\mathbf{z}\,\mathrm{Tr}\left( \rho A_{\mathbf{z}'}\right)
  \mathrm{Tr}\left( \mathcal{E}^{\dagger}(A_{\mathbf{z}})A_{\mathbf{z}'}\right)  \\
  & =(\pi \hbar)^n\int \mathrm{d}\mathbf{z}'\, W_{\rho}(\mathbf{z}')
  \mathrm{Tr}\left( A_{\mathbf{z}}\mathcal{E}(A_{\mathbf{z}'})\right) \,.
\end{align}
Eq.~(\ref{eq:Wigneractionofmap}) then follows.

\subsubsection{Gaussian quantum mechanics}

We define Gaussian states, measurements and transformations in terms of their Wigner representations. The Gaussian \emph{states} are the $\rho$ for which
\begin{equation}
  W_{\rho}(\mathbf{z})=W_{\rho}(0)e^{-\frac{1}{2}(\mathbf{z}-
  \mathbf{d})^{T}\gamma ^{-1}(\mathbf{z}-\mathbf{d})}\,,
\end{equation}
where $\gamma$ is the covariance matrix of $\rho$ and $\mathbf{d}$ is the vector of its means.\footnote{It is more common to define the Gaussian quantum states as those for which the Wigner-characteristic function is Gaussian, but because the latter is the Fourier transform of the Wigner representation of the state, the two definitions are equivalent.}  Gaussian states have Wigner functions that are positive everywhere; note that the only pure quantum states with positive Wigner functions are the Gaussian pure states~\cite{Hud74}.  Next, we can define the Gaussian measurements as those which, implemented on one half of a
system in a Gaussian state, necessarily leave the other half in a Gaussian state as well.  We can also define the Gaussian operations as those which implemented on a system or part of a system in a Gaussian state, take the system to another Gaussian state.  Doing so, one finds that the Gaussian measurements and transformations are those whose Wigner representations are Gaussian.  Specifically, the Gaussian measurements are the POVMs $\{E_{\mathbf{y}}\}$ for which we have
\begin{equation}
  W_{E_{\mathbf{y}}}(\mathbf{z})=W_{E_{\mathbf{y}}}(0)e^{-\frac{1}{2}(\mathbf{z}-\mathbf{d}_{\mathbf{y}})^{T}\gamma _{\mathbf{y}}^{-1}(\mathbf{z}-\mathbf{d}_{\mathbf{y}})}\qquad \forall\ \mathbf{y}\,,
\end{equation}
where $\gamma _{\mathbf{y}}$ is the covariance matrix of $E_{\mathbf{y}}$ and $\mathbf{d}_{\mathbf{y}}$ is the vector of its means, and for which $\int \mathrm{d}\mathbf{y} W_{E_{\mathbf{y}}}(\mathbf{z}) =1$ for all $\mathbf{z}$.
The Gaussian transformations are the CP maps $\mathcal{E}$ for which we have
\begin{equation}
  W_{\mathcal{E}}(\mathbf{z}|\mathbf{z}^{\prime })=W_{\mathcal{E}}(0|0)e^{-\frac{1}{4}\mathbf{z}^{\prime \prime T}\gamma _{\mathcal{E}}^{-1}\mathbf{z}^{\prime \prime }+d_{\mathcal{E}}^{T}\mathbf{z}^{\prime \prime }}
\end{equation}
where $\mathbf{z}'' \equiv (\mathbf{z},\mathbf{z}')$.

%\section{Operational equivalence of ERL mechanics and Gaussian quantum mechanics}
\subsection{Proof of Equivalence}
\label{sec:proofmain}

We can now provide the proof of theorem \ref{thm:Gaussian}.

Note first that the Wigner representation of a Gaussian state
can be interpreted as a probability distribution on phase-space.
This is because it is both positive, by virtue of the fact that a
Gaussian distribution is positive, and normalized to unity, which
one verifies by noting that $\int \mathrm{d}\mathbf{z}\, W_{\rho}(\mathbf{z})=\mathrm{Tr}(\rho\int
\mathrm{d}\mathbf{z}\,A_{\mathbf{z}})=\mathrm{Tr}(\rho I)=1$.

Note further that the Wigner representation $\{W_{E_{\mathbf{y}}}(\mathbf{z})\}$ of a POVM $\{E_{\mathbf{y}}\}$ can be interpreted
as a set of conditional probabilities for the outcome to lie within $\textrm{d}\mathbf{y}$ of ${\mathbf{y}}$ given that the ontic state is $\mathbf{z}$.  Again, positivity follows from Gaussianity. The fact that the $W_{E_{\mathbf{y}}}(\mathbf{z})$ form a probability density over $\mathbf{y}$ for all $\mathbf{z}$ becomes evident when one notes that $\int \textrm{d}\mathbf{y} W_{E_{\mathbf{y}}}(\mathbf{z})=(\pi \hbar)^n\int \textrm{d}\mathbf{y} \mathrm{Tr}(A_{\mathbf{z}}E_{\mathbf{y}})= (\pi \hbar)^n\mathrm{Tr}(A_{\mathbf{z}}I)=1$ for all $\mathbf{z}.$

Furthermore, it follows from Eq.~(\ref{eq:Bornrule}) that one can interpret the probability density of obtaining outcome $\mathbf{y}$ in a measurement associated with a Gaussian POVM $\{E_{\mathbf{y}} \}$ upon a Gaussian state $\rho$ as the probability density of obtaining outcome $\mathbf{y}$ given $\mathbf{z}$ weighted by the probability density of $\mathbf{z}$.  In other words, the preparation procedure associated with a Gaussian state $\rho$ can be understood as the preparation of a system at some unknown point $\mathbf{z}$ in phase space, with probability distribution $W_{\rho}(\mathbf{z})$, and the measurement procedure associated with a Gaussian POVM $\{E_{\mathbf{y}}\}$ can be understood as revealing information about $\mathbf{z}$ by the fact that different $\mathbf{z}$ may vary in the probability densities they assign to the different outcomes.

Finally, we can interpret the Wigner representation $W_{\mathcal{E}}(\mathbf{z}|\mathbf{z}')$ of a Gaussian trace-preserving CP map $\mathcal{E}$ as a conditional probability of $\mathbf{z}$ given $\mathbf{z}'$ (thereby justifying the choice of notation).  Positivity of $W_{\mathcal{E}}(\mathbf{z},\mathbf{z}')$ follows from its Gaussianity, and the fact that $\int \mathrm{d}\mathbf{z}\, W_{\mathcal{E}}(\mathbf{z},\mathbf{z}')=1$ for all $\mathbf{z}'$ is verified by noting that $\int \mathrm{d}\mathbf{z}\, \mathrm{Tr}\left( A_{\mathbf{z}}\mathcal{E}(A_{\mathbf{z}^{\prime }})\right) =\mathrm{Tr}\left(
\mathcal{E}(A_{\mathbf{z}^{\prime }})\right) =\mathrm{Tr}\left( A_{\mathbf{z}^{\prime }}\right) =1$ where the second identity is due to the assumption that $\mathcal{E}$ is trace-preserving.  Thus, Eq.~(\ref{eq:Wigneractionofmap}) can be interpreted as follows.  If the initial distribution over phase space is $W_{\rho}(\mathbf{z}'),$ and the probability of $\mathbf{z}'$ being mapped to $\mathbf{z}$ is $W_{\mathcal{E}}(\mathbf{z},\mathbf{z}^{\prime }),$ then the final distribution over phase space is $\int\mathrm{d}\mathbf{z}'\, W_{\mathcal{E}}(\mathbf{z},\mathbf{z}')W_{\rho}(\mathbf{z}')$.

We have seen, therefore, that the Wigner representation of Gaussian quantum mechanics yields the same sorts of descriptions of preparations, measurements and transformations that one finds in Liouville mechanics.  But are they precisely the subset picked out by our epistemic constraint?  Yes.  To demonstrate this, we need only show that the Wigner representations satisfy the conditions of the classical uncertainty principle (with $\lambdabar$ replaced with $\hbar$), that is, the conditions implied by demanding that the phase-space distributions satisfy the classical version of the Heisenberg uncertainty relation.

From Eq.~(\ref{eq:CanonicalUncertaintyRelations}), we have that any quantum state $\rho$ has a covariance matrix $\gamma(\rho)$ that satisfies the uncertainty relation $\gamma (\rho)+i\hbar \Sigma \geq 0$.  To relate this result to the Wigner function, we require the following lemma:

\begin{lemma}
The covariance matrix $\gamma(\rho)$ of a quantum state $\rho$ (defined in terms of quantum expectation values $\left\langle f\right\rangle_{\rho}$) is equal to the covariance matrix $\gamma (W_{\rho})$ of its Wigner function $W_{\rho}$, considered as a function over phase-space (defined in terms of classical expectation values $\left\langle f\right\rangle_{W_{\rho}}=\int \mathrm{d}\mathbf{z}\,W_{\rho}(\mathbf{z})f(\mathbf{z})$).
\end{lemma}

\begin{proof}
All moments of the Wigner function are given by the expectation values of symmetrically-ordered products of the canonical operators.  See, for example, Ref.~\cite{GardinerZoller} for a proof; here, we reproduce this result in detail for the first two moments.  Recall the definition of the covariance matrix of $\rho$, Eq.~(\ref{eq:CovarianceMatrix}).  We wish to rewrite this in
the Wigner representation.  First, note that
\begin{align}
  \gamma (\rho) &= 2\mathrm{Re}\,\mathrm{Tr}\left( \rho(\hat{z}_{i}-\xi _{i})(\hat{z}_{j}-\xi _{j})\right)
  \nonumber \\
  &=\,\mathrm{Tr}\left( \rho\left[ (\hat{z}_{i}-\xi _{i})(\hat{z}_{j}-\xi _{j})
  +(\hat{z}_{j}-\xi _{j})(\hat{z}_{i}-\xi _{i})\right]\right)\, .
\end{align}
The Wigner representations of quadratic observables are
\begin{align}
  W_{\hat{z}_{i}}(\mathbf{z}) &= z_{i}, \\
  W_{\hat{z}_{i}\hat{z}_{j}+\hat{z}_{j}\hat{z}_{i}}(\mathbf{z})
  &=2z_{i}z_{j}\,.
\end{align}
This follows from direct evaluation of the Gaussian integrals in Eq.~(\ref{eq:WignerObservables}).

Thus, first-order moments of $\rho$ coincide with those of $W_{\rho}(\mathbf{z})$,
\begin{equation}
  \xi_i =\left\langle \hat{z}_{i}\right\rangle _{\rho}
  =\int \mathrm{d}\mathbf{z}\,z_{i}W_{\rho}(\mathbf{z})
  =\left\langle \hat{z}_{i}\right\rangle_{W_{\rho}(\mathbf{z})} \,,
\end{equation}
and the covariance matrix of the quantum state $\rho$ coincides with
that of the phase-space function $W_{\rho}(\mathbf{z})$,
\begin{align}
  \gamma (\rho) &=2\int \mathrm{d}\mathbf{z}(z_{i}-\xi_{i})(z_{j}-\xi_{j})W_{\rho}(\mathbf{z}) \nonumber \\
  &=\gamma (W_{\rho})\,.
\end{align}
\end{proof}

Therefore, the Wigner representations of Gaussian states satisfy the classical uncertainty principle, with $\lambdabar = \hbar$.  All that remains is to show that the Wigner representations of Gaussian measurements and transformations coincide with the phase-space representations of measurements and transformations in ERL mechanics.

Any POVM element $E$, when normalized, can be viewed as a density operator.  Thus, the covariance matrix of a normalized POVM element, $\gamma (E/\mathrm{Tr}(E))$, satisfies the Heisenberg uncertainty relation, $\gamma (E/\mathrm{Tr}(E))+i\hbar \Sigma \geq 0$.  This implies that the normalized Wigner function $W_{E}(\mathbf{z})/|W_{E}(\mathbf{z})|$ where $|W_{E}(\mathbf{z})|=\int \mathrm{d}\mathbf{z}'\,W_{E}(\mathbf{z})$ satisfies the classical uncertainty relation, and thus $W_{E}(\mathbf{z})$ is a valid indicator function.

\begin{equation}
W_{\Psi^{\text{corr}}}(\mathbf{z}_A,\mathbf{z}_B) \propto \prod_i \delta(q_{iA}-q_{iB})\delta(p_{iA}+p_{iB})\,,
\end{equation}
corresponds to a Gaussian state $| \Psi^{\text{corr}} \rangle \langle \Psi^{\text{corr}}|$.  Necessity follows from the Choi-Jamiolkowski isomorphism in Gaussian quantum mechanics, which ensures that transformations $\mathcal{E}$ that lead to a Gaussian bipartite state when acting on the perfectly correlated state in $\mathcal{M} \times \mathcal{M}$,
\begin{equation}
  \rho =(\mathcal{E}\otimes \mathcal{I}) | \Psi^{\text{corr}} \rangle \langle \Psi^{\text{corr}}|\,,
\end{equation}
are necessarily Gaussian transformations.

With this, we have proved Theorem~\ref{thm:Gaussian}.

\section{Conclusions}

In the introduction, we emphasized that ERL mechanics can reproduce a large number of quantum phenomena. We have explained at length how it does so for several important examples.  These phenomena can therefore be understood intuitively in terms of a simple story about uncertainty in a classical world. Given that ERL mechanics is operationally equivalent to Gaussian quantum mechanics, if a phenomena exists in Gaussian quantum mechanics, then we are assured that it exists within ERL mechanics and that such a story can be provided, even if we do not bother to extract it from the formalism.  Therefore, to know the explanatory scope of ERL mechanics,  it suffices to determine which quantum phenomena are found within Gaussian quantum mechanics.  Fortunately, much work has already been done in determining what aspects of quantum theory, in particular, what aspects of quantum information theory, are present in Gaussian quantum mechanics, and so we simply refer the reader to this work.  The list includes: basic phenomena of quantum theory such as the no-cloning theorem~\cite{Cer00}, the Einstein-Podolsky-Rosen effect~\cite{Reid09} and quantum teleportation~\cite{Bra98}; information-processing tasks such as dense coding~\cite{Bra00}, quantum key distribution~\cite{Ral00}, and quantum error correction~\cite{Bra98b}; and many aspects of entanglement theory~\cite{Gie03,Eis03}.  For a review of the subject of information theory using continuous-variable systems and Gaussian quantum mechanics, see Ref.~\cite{Bra05,Wee11}.

The classical theory that we have used as our starting point has been particle mechanics. However, we could have equally well considered any degrees of freedom described by canonical coordinates on a symplectic vector space.  In particular, we could have considered fields.  Indeed, the most significant application of Gaussian quantum mechanics is to quantum optics.  One can interpret the theory proposed here as a classical statistical theory of optics with an epistemic restriction which is operationally equivalent to the subtheory of quantum optics which consists of Gaussian states, measurements and transformations.  The set of Gaussian states is the set of all coherent states (including the vacuum state), all squeezed states (including quadrature eigenstates) and all multimode versions of these.  The Gaussian transformations are those that can be achieved using the standard toolkit of optical elements -- beam splitters, phase shifters and squeezers -- as well as linear attenuation and amplification.  The Gaussian measurements can all be constructed from a homodyne detection preceded by one of the above transformations\footnote{Note, however, that this does not include direct photodetection.}.  All experiments in quantum optics that make use of only these elements can therefore be furnished with an intuitive explanation in terms of statistical optics with an epistemic restriction; for instance, such a description of the quantum teleportation experiment of Ref.~\cite{Fur98} is provided in Ref.~\cite{Cav04}.

%In future work, we could investigate the result of applying an epistemic restriction to the statistical theory of classical dipoles interacting with scalar fields.  How many of the quantum phenomena of quantum electrodynamics might be reproduced in this way?

The explanatory scope of ERL mechanics adds further credibility to the research program wherein the quantum state is taken to be a representation of an agent's incomplete knowledge of reality rather than a representation of reality itself. A skeptic might challenge the notion that our results constitute interpretational progress on the grounds that the mystery of quantum theory has just been shifted to the question: why the epistemic restriction?  We have several responses to this charge.  First, any progress in reconstructing quantum theory from simple principles holds interpretational lessons, even if further elucidation and justification of the principles is required.  Second, and more importantly, we feel that the interpretation of Gaussian quantum mechanics in terms of ERL mechanics is more compelling than most competing interpretations, for instance, Everett's~\cite{Everett} or the one of de Broglie and Bohm~\cite{Bohmian}, because the latter interpretations are mathematically inspired -- they start from the mathematical formalism of quantum theory and attempt to tell an ontological story that does justice to this formalism, while the reconstruction provided here is conceptually inspired -- we start with a classical ontology that is conceptually unproblematic, add the conceptual innovation of an epistemic restriction, and \emph{derive} the mathematical formalism of Gaussian quantum mechanics.  Third, we feel that the approach described here succeeds at unscrambling  Jaynes' omelette of ontological and epistemological notions in a more satisfying fashion than other approaches.

Of course, although the length of the list of quantum phenomena that are reproduced by ERL mechanics is long, it is not complete. Neither ERL mechanics, nor any classical statistical theory with an epistemic restriction, can do justice to Bell's theorem or the Kochen-Specker theorem.  We must grow the list of such outstanding phenomena and focus upon them for it is these that will dictate what other conceptual innovations are required to reproduce the full quantum theory within this program.

The relation between Gaussian quantum mechanics and ERL mechanics is strongly analogous to the relation that exists between the stabilizer theory for qutrits \cite{Got01} and a classical statistical theory of trits with an epistemic restriction (trits are three-level classical systems and qutrits are three-level quantum systems). The latter sort of theory, which makes use of a classical phase space over a discrete field, has been developed in Refs.~\cite{SchreiberSpekunpublished,EdwardsCoeckeSpekkens}. The proof that it is operationally equivalent to the stabilizer theory for qutrits proceeds by showing that it reproduces the discrete Wigner representation for odd-dimensional systems that was proposed by Gross \cite{Gross} (which is positive on stabilizer states).  Just as it is well-known in quantum information circles that stabilizer states are the natural discrete analogues of Gaussian states, the classical statistical theory of trits with an epistemic restriction of Ref.~\cite{SchreiberSpekunpublished} is the natural discrete analogue of ERL mechanics.

One is naturally led to ask whether one can find a similar relation between the stabilizer theory for qu\emph{bits} and and a classical statistical theory of bits with an epistemic restriction.  The latter sort of theory has been developed by Ref.~\cite{Spe07} and is commonly known as the ``Spekkens Toy Theory''.  One finds that in this case the two theories in question are \emph{not} operationally equivalent.  Such inequivalence is inevitable because a classical theory with an epistemic restriction is \emph{by construction} a local noncontextual hidden variable theory and it is known that one can prove Bell's theorem and the Kochen-Specker theorem within the stabilizer theory of qubits (for instance, by using the GHZ version of Bell's theorem~\cite{GHZ}).  For the case of Gaussian quantum mechanics and the stabilizer theory of qu\emph{trits}, the fact that one \emph{can} reconstruct these from a restriction upon a classical statistical theory shows that one \emph{cannot} prove Bell's theorem or the Kochen-Specker theorem within these subtheories.

Another question that arises naturally is whether we might be able to find another epistemic restriction that yields a theory which is more comprehensive than ERL mechanics, that is, one that is operationally equivalent to a subtheory of quantum mechanics that has a larger scope than Gaussian quantum mechanics. 
As it turns out, if one demands that this larger subtheory \emph{includes} Gaussian quantum mechanics, then the question has a negative answer.  The reason is that there is no subtheory of quantum mechanics that is ``between'' Gaussian quantum mechanics and the full theory, so there is nothing to shoot for in such a reconstruction.  To be precise, it has been shown~\cite{Lloyd99,Bar02} that if one adds to the set of unitaries allowed in Gaussian quantum mechanics (those generated by quadratic Hamiltonians) even a \emph{single} unitary from outside this set and then closes under composition, one obtains \emph{all} unitaries. An analogous result is widely believed to hold (but to our knowledge has not been rigorously proven) for the stabilizer theory of qudits: if one adds \emph{any} additional unitary to those allowed within the stabilizer theory, commonly known as the \emph{Clifford group}, and closes under composition, one obtains all unitaries over the qudits. (This question can be rephrased in the language of quantum computation as a question about the universality of a gate set \cite{Bar95}.)  In other words, ``next stop: quantum theory''.

\begin{acknowledgments}
We thank Sarah Croke, Andrew Doherty, Matthew Palmer, Roberta Rodriquez and Olaf Schreiber for discussions, Robin Blume-Kohout for helping us to recognize the need for entropy maximization in the epistemic restriction, and Giulio Chiribella, Paolo Perinotti and Caslav Brukner for a useful discussion on the classical analogue of no-cloning.
SDB acknowledges support from the Australian Research Council.  Part of this work was completed
while RWS was at the University of Cambridge where he was supported by the Royal Society as an international research fellow.  Research at Perimeter Institute is supported in part by the Government of Canada through NSERC and by the Province of Ontario through MRI.  TR acknowledges the support of the UK Engineering and Physical Sciences Research Council.
\end{acknowledgments}

\appendix

\section{Motivation for the max-ent condition}
\label{sec:Whymaxent}

In Sec.~\ref{sec:EpistemicRestriction}, we noted that one of the reasons for incorporating the max-ent condition into the epistemic constraint is that without it, one obtains a much smaller set of valid measurements\footnote{We thank Robin Blume-Kohout for pointing this out to us.  The demonstration we provide is a modification of one that he suggested.}.  We are now in a position to see why this is the case.  Imagine a distribution over a single system of the form
\begin{multline}
\mu_{\textrm{test}}(q_A,p_A) \propto
( \frac{1}{2} G_{-q_0,\delta q} (q_A) +  \frac{1}{2} G_{q_0,\delta q} (q_A))
 \\ \nonumber
\times ( \frac{1}{2} G_{-p_0,\delta p} (p_A) +  \frac{1}{2} G_{p_0,\delta p} (p_A))  \\ \nonumber
\end{multline}
where $q_0 \gg \delta q \gg \lambdabar/p_0$, and $p_0 \gg \delta p \gg \lambdabar/q_0$.  This satisfies the CUP because the variances are $\Delta q_A \simeq q_0$ and $\Delta p_A \simeq p_0$, such that $\Delta q_A \Delta p_A \gg \lambdabar$ (there are no cross-correlations). However, because $\mu_{\textrm{test}}$ is not a multi-variate Gaussian it violates the max-ent condition.  We will be interested in the limiting case where $\delta q, \delta p \to 0$ and $q_0,p_0 \to \infty$.

If there is no max-ent condition, then for a pair of systems we would also have to allow a distribution of the form
\begin{multline}
\mu'_{\textrm{test}}(q_A,p_A,q_B,p_B) \\ \nonumber
\propto \mu_{\textrm{test}}(q_A,p_A)
G_{0,\delta q} (q_A-q_B ) G_{0,\delta p} (p_A + p_B )
%\mu_{\textrm{corr}}(q_A,p_A,q_B,p_B)
%G_{0,\delta q} (q_A-q_B) ( \frac{1}{2} G_{p_0, \delta p} (p_A + p_B)  \\ \nonumber
%+ \frac{1}{2} G_{-p_0, \delta p} (p_A + p_B) ),
\end{multline}
Note that $q_A$ and $q_B$ are strongly correlated (positively) if $\delta q$ is small and $p_A$ and $p_B$ are strongly correlated (negatively) if $\delta p$ is small. The correlation is perfect in the limit that $\delta q, \delta p \to 0$. The distribution $\mu'_{\textrm{test}}$ also satisfies the CUP but not the max-ent condition.

Now imagine a Gaussian indicator function on $A$ that has variances $\Delta q_A \simeq \delta' q$ and $\Delta p_A \simeq \delta' p$ where $q_0 \gg \delta' q \gg \delta q$ and $p_0 \gg \delta' p \gg \delta p$.
%which is of the order of $\delta' p$ (the value of which will be specified in a moment), and a variance in $q_A$ of the order of $\lambdabar/\delta' p$.
Because we can take the limiting case of $\delta q, \delta p \to 0$ and $q_0,p_0 \to \infty$, these inequalities place no constraint on $\delta' p$, $\delta' q$, so we can consider an arbitrary Gaussian indicator function on $A$ that satisfies Eq.~\eqref{eq:validindicatorfunctions}. It is not too difficult to see that if distributions of the form of $\mu'_{\textrm{test}}$ were allowed, then every such indicator function will be ruled out.

The argument is by contradiction.  We show that if an indicator function satisfying Eq.~\eqref{eq:validindicatorfunctions} were allowed, then it would imply a violation of the CUP.
 We assume that the initial state of $AB$ is $\mu'_{\textrm{test}}$ where $\delta q, \delta p \to 0$ and $q_0,p_0 \to \infty$.  First, note that for such a state, $q_A$ prior to the measurement is arbitrarily close in value to either $q_0$ or $-q_0$.  Furthermore, given that we have chosen
$q_0 \gg \delta' q \gg \delta q$, a measurement of an indicator function with $\Delta q_A \simeq \delta' q$  can reveal which value $q_A$ takes with arbitrarily high accuracy.
% First, note that a measurement of such an indicator function can provide arbitrarily high certainty about whether the value of $q_A$ prior to the measurement was $q_0$ or $-q_0$ (these are the only two values it can take according to $mu'_{\textrm{test}}$).
By virtue of the arbitrarily strong correlation between $q_A$ and $q_B$ in $\mu'_{\textrm{test}}$ and the lack of any influence from $A$ to $B$, one would thereby learn with arbitrarily high accuracy what the value of $q_B$ was after the measurement (whether it is close in value to $q_0$ or to $-q_0$).  Meanwhile, such a measurement could also distinguish with arbitrarily high accuracy whether $p_A$ had a value close to $-p_0$ or to $p_0$ prior to the measurement and again by virtue of the arbitrarily strong correlation between $p_A$ and $p_B$ and the lack of any influence from $A$ to $B$, one would thereby learn with arbitrarily high certainty what the value of $p_B$ was after the measurement (whether it is close in value to $p_0$ or to $-p_0$). Consequently, if such a measurement were allowed, one would be able to infer both the values of $q_B$ and $p_B$ after the measurement to arbitrary accuracy.  Because this would violate the CUP part of the epistemic constraint, such a measurement would have to be ruled out.  Therefore, if we relaxed the max-ent condition, then the resulting theory would include none of the indicator functions that are included in ERL mechanics.

\end{document}